\documentclass[11pt]{article}

\usepackage[letterpaper,
			left = 1.1truein,  
			right = 1.1truein, 
			top = 1.1truein, 
			bottom = 1.1truein]{geometry}

\usepackage{subcaption}
\usepackage{algpseudocode,amssymb,mathtools,amsmath,amsthm}
\usepackage{algorithmicx}
\usepackage{algorithm}
\usepackage{arydshln,mathabx}
\usepackage{xcolor}
\algnewcommand\algorithmicoutput{\textbf{Output:}}
\algnewcommand\Output{\item[\algorithmicoutput]}
\usepackage{bbm}
\usepackage{mathrsfs}
\allowdisplaybreaks

\usepackage[%dvips,
            CJKbookmarks=true,
            bookmarksnumbered=true,
            bookmarksopen=true, 
            colorlinks=true,
            citecolor=red,
            linkcolor=blue,
            anchorcolor=red,
            urlcolor=blue
            ]{hyperref}

\newtheorem{defn}{Definition}[section]
\newtheorem{thm}{Theorem}[section]
\newtheorem{lem}{Lemma}[section]
\newtheorem{cor}{Corollary}[section]

\newtheorem{rem}{Remark}[section] 
\newcommand{\E}{\mathbb{E}}
\newcommand{\KL}{\mathrm{KL}}
\renewcommand{\P}{\mathbb{P}}
\newcommand{\M}{\mathcal{M}}
\newcommand{\pM}{M_{{\mathscr G},\varepsilon}}
\newcommand{\pW}{W_{{\mathscr G},\varepsilon}}

\usepackage{enumitem}
\setlist[enumerate]{leftmargin=.5in}
\setlist[itemize]{leftmargin=.5in}

\DeclareMathOperator*{\argmax}{arg\,max}

\usepackage[numbers, square, sort]{natbib} 
\usepackage{lipsum}

\newcommand{\R}{\mathbb R}
\newcommand{\pL}{L_{{\mathscr G},\varepsilon}}

\title{PriME: Privacy-aware Membership profile Estimation in networks}

\author{
Abhinav Chakraborty\footnote{Email: ac4662@columbia.edu} \\
\emph{Columbia University}
\and
Sayak Chatterjee\footnote{Email: sayakc@wharton.upenn.edu} \\
\emph{University of Pennsylvania}
\and
Sagnik Nandy\footnote{Email: nandy.15@osu.edu} \\
\emph{The Ohio State University}
}
\begin{document}
	\maketitle

\begin{abstract}
This paper presents a novel approach to estimating community membership probabilities for network vertices generated by the Degree Corrected Mixed Membership Stochastic Block Model \citep{JIN2023} while preserving individual edge privacy. Operating within the $\varepsilon$-edge local differential privacy framework, we introduce an optimal private algorithm based on a symmetric edge flip mechanism and spectral clustering for accurate estimation of vertex community memberships. We conduct a comprehensive analysis of the estimation risk and establish the optimality of our procedure by providing matching lower bounds to the minimax risk under privacy constraints. To validate our approach, we demonstrate its performance through numerical simulations and its practical application to real-world data. This work represents a significant step forward in balancing accurate community membership estimation with stringent privacy preservation in network data analysis. 

\end{abstract}
\section{Introduction}

Community detection in networks has found important applications in the study of social networks, image segmentation, and cell biology (see \cite{Fortunato_2010} for a survey). \emph{Stochastic Block Model} (SBM), introduced by \cite{Holland1983StochasticBF}, is a probabilistic generative model for random networks with latent community structures that provides a rigorous framework for the analysis of the performance of community recovery methods in such networks. These applications often involve working with sensitive datasets in several disciplines like healthcare, social networks, and finance, where there is an inherent risk of extracting sensitive information about individuals through adversarial queries or analysis \cite{houghton2014privacy,zheleva2011privacy,zhou2008preserving}. Given the sensitive nature of the data encountered in these applications, ensuring privacy and data protection becomes paramount. While extensive literature has been dedicated to exploring the theory and performance of various community recovery algorithms applied to Stochastic Block Models (SBM) (see \cite{JMLR:v18:16-480,10.1214/19-STS736} for comprehensive surveys), relatively fewer of these studies \cite{chen2023private,Hehir_Slavkovic_Niu_2022} have delved into this problem while considering the crucial aspect of privacy constraints.

In this paper, we study a generalized version of SBM, where nodes are allowed to have membership in multiple communities. The community membership profile of each node in the random network is characterized by a probability vector with the entries encoding the likelihood of the node belonging to the particular community. We propose a novel approach to estimating such membership profiles in random networks maintaining privacy constraints. In different practical scenarios, the assumption that nodes exclusively belong to one cluster becomes unreasonable, highlighting the need to embrace heterogeneity in community membership studied in this paper. Such modeling of heterogeneity in community membership of the network nodes through Mixed Membership Models has gained popularity in the literature recently and has been successfully used in the study of citation networks, topic modeling, image processing, and transcriptional regulation \cite{10.5555/2073876.2073918,10.5555/944919.944937,1467486,doi:10.1073/pnas.0307760101, NIPS2008_8613985e,doi:10.1080/07350015.2021.1978469} among many others. However, to the best of our knowledge, algorithms to estimate membership profiles in such models maintaining privacy constraints have not been studied before.

We consider $K$ communities $\mathcal C_1,\ldots,\mathcal C_K$ and an observed network $\mathscr G$ on $n$ vertices. The probability of each edge $(i,j) \in {\binom{n}{2}}$ being selected in the edge set $E(\mathscr G)$ of $\mathscr G$ is given by
\begin{align}
\label{eq:adj_matrix}
\P[(i,j) \in E(\mathscr G)] = \Pi^\top_{i*} B \Pi_{j*}
\end{align}
where $\Pi \in \R^{n \times K}$, $B \in \R^{K \times K}$ represent the membership and the connection probability matrix respectively. The notation $\Pi_{i*}$ refers to the $i$-th row of $\Pi$. For $i \in \{1,\ldots,n\}$ and $k \in \{1,\ldots,K\}$, the entries $\Pi_{ik}$ of $\Pi$ represent the likelihood of node $i$ belonging to community $\mathcal C_k$. The entry $B_{ab}$ for $(a,b) \in\{1,\,2,\,\ldots,\,K\}\times\{1,\,2,\,\ldots,\,K\}$ represents the probability of an edge between the clusters $\mathcal C_a$ and $\mathcal C_b$. This modeling of the random network with mixed community identity of the nodes is called \emph{Mixed Membership Stochastic Block Model} (MMSBM) and has been studied extensively in \cite{JIN2023,jin2017sharp,ke2022optimal}. For example, in \cite{JIN2023} the authors propose a spectral method called MIXED-SCORE to estimate the membership matrix $\Pi$ in a degree-corrected version of MMSBM where
\begin{equation}
\label{eq:DCMM}
  \P[(i,j) \in E(\mathscr G)] = \theta_i\theta_j\Pi^\top_{i*} B \Pi_{j*} 
\end{equation}
where the parameters $\theta_1,\,\theta_2,\,\ldots,\,\theta_n$ represent the degree heterogeneity of the nodes.  In other words, higher values of $\theta$ denote a higher number of connections incident to the node.

\begin{rem}
\label{rem1_1}
Observe that the membership profile matrix $\Pi$ is generally not identifiable in the model \eqref{eq:DCMM} under general conditions. A sufficient condition to ensure identifiability is the non-singularity of the matrix $B$, along with the requirements that $B_{ii} = 1$ for all $i \in [n]$ and that each community contains at least one pure node with a significantly large degree whose membership profile assigns full weight to the specific community. The proof of this property is provided in Proposition A.1 of \citet{JIN2023}. Even under these conditions, $\Pi$ is identifiable only up to a permutation of its columns. However, this ambiguity corresponds merely to a relabeling of the communities, which is acceptable for all practical purposes.
\end{rem}

The authors in \cite{JIN2023} consider the estimation loss given by
\begin{align}
\label{eq:loss}
  \mathcal{L}(\widehat \Pi,\Pi)=\min_{T}\bigg\{\frac{1}{n}\sum_{i=1}^{n}\|(T\widehat{\Pi})_{i*}-\Pi_{i*}\|_1\bigg\}  
\end{align}
with the minimum taken over all possible permutations of the community labels and prove that the estimation loss of MIXED-SCORE converges to $0$ at a $n^{-1/2}$ rate. 
%In \cite{jin2017sharp} the authors show that this rate is, in fact, minimax optimal under the assumption $\theta_{\max} \le C\theta_{\min}$ for some $C>0$, where $\theta_{\max}$ and $\theta_{\min}$ are the maximum and minimum values of $\theta$. 
In \citet{ke2022optimal}, the authors prove that this rate is minimax optimal using a method that combines SCORE normalization with spectral clustering. It is also worthwhile to mention that some other methods have also been considered to recover communities in MMSBM like a spectral method based on $K$-medians clustering in \cite{article_levina} and a Bayesian approach using a Dirichlet prior on the rows of $\Pi$ in \cite{NIPS2008_8613985e}.

In this paper, we focus on ensuring \emph{Local Differential Privacy} (Local DP) (see  \cite{yang2020local} for a survey) in estimating $\Pi$ from $\mathscr G$. This distinguishes itself from \emph{central Differential Privacy} (central DP) (first introduced in \cite{dwork2006differential}) by eliminating the reliance on a trusted curator, providing us with a more robust privacy guarantee, reducing the risk of data leakage and unauthorized disclosure of information. The notion of Local DP is particularly relevant to community detection problems. For example, consider a healthcare network where patient data is shared for research purposes \cite{abouelmehdi2018big}. Local DP ensures that the privacy of individual patients is maintained, preventing the leakage of sensitive medical conditions or personal information. Similarly, in social networks and financial transactions, Local DP safeguards user identities and affiliations, thwarting potential threats and breaches.

\subsection{Related Literature and our contributions}
\emph{Differential Privacy}(DP) introduced in \citet{dwork2006calibrating}, provides a rigorous framework to mathematically quantify data privacy in statistical procedures. Over the last decade, significant research has focused on strengthening the algorithmic and mathematical foundations of differential privacy \cite{dwork2011firm,dwork2014algorithmic,doi:10.1198/jasa.2009.tm08651} leading to different procedures for differentially private statistical learning \cite{2016arXiv160301887D, 10.1007/978-3-662-53641-4_24, 10.1145/3188745.3188946, 10.1111/rssb.12454,Balle2018ImprovingTG,nissim2007smooth}. 

Delineating the trade-off between privacy requirements and the accuracy of a statistical procedure is a flourishing area of research in mathematical statistics. A substantial body of literature has explored statistical procedures having the best possible accuracy (minimax optimal risk) under stringent privacy constraints \cite{doi:10.1137/15M1033587,Barber2014PrivacyAS,10.1214/21-AOS2058,avella2019differentially,kamath2022private,pmlr-v99-kamath19a}. 

Maintaining data privacy is extremely important in analyzing network data, since such relational data often reveal sensitive information about the subjects under study. In recent years, there has been a surge in the literature focused on rigorously formulating the idea of \emph{differential privacy} in the context of graphical data \cite{nissim2007smooth,kasiviswanathan2013analyzing,nguyenvulli,guo2023privacy,xu2023binary} and developing private estimation procedures for parameters such as centrality measures, clustering coefficients, and other network structures \cite{wang2013learning, LAEUCHLI2022126546,ning2021benchmarking,houghton2014privacy,zheleva2011privacy,borgs2015private,borgs2018revealing}. We refer to \citet{mueller2022sok} for a survey of related work. 

The development of private procedures for recovering the latent community structure in an observed network utilizing the density of connections between different vertices of the network is an important area of research. The existing methods to privately recover such underlying communities in graphical data fall into three main categories: perturbing non-private estimators of community labels by adding random noise, Bayesian sampling mechanisms, and adding noise to the network itself before applying spectral clustering to recover the latent communities. We refer to \citet{chen2023private,mohamed2022differentially,guo2023privacy,Hehir2021ConsistencyOP} for a comprehensive survey of these methods and their theoretical properties. 

Though previous research has explored the intersection of Stochastic Block Models (SBM) and privacy concerns, such investigations have not addressed the unique problem of recovering community membership profiles in mixed membership models while adhering to the Local DP constraints.
In this paper, we consider the setting defined by \eqref{eq:DCMM}, where the network vertices belong to multiple communities. We characterize the fundamental tradeoff between privacy constraints and the statistical utility of any algorithm used to infer the membership profile matrix $\Pi$ from privatized data by establishing the minimax estimation rate of the parameter $\Pi$ under the loss $\mathcal{L}$ defined in \eqref{eq:loss}.

Additionally, we propose a specific mechanism to modify the observed data to ensure $\varepsilon$-local differential privacy and introduce a spectral algorithm, PriME, that estimates the membership profile matrix $\Pi$ while adhering to user-specified local differential privacy constraints. We demonstrate that, under the loss $\mathcal{L}$, the estimator produced by PriME is consistent in estimating $\Pi$ and achieves the optimal minimax rate. To the best of our knowledge, this work is the first to establish the fundamental tradeoff between statistical accuracy and privacy in the context of community detection in networks with mixed membership profiles across multiple communities.

This uncharted territory underscores the novelty and significance of our work in developing algorithms at the intersection of mixed membership modeling of networks with latent communities and privacy protection.

On the technical front, we utilize information processing inequalities to characterize the minimax lower bound for estimating $\Pi$ from privatized data. Our analysis departs significantly from traditional techniques used to establish minimax lower bounds in mixed membership problems, as our lower bound considers not only the class of estimators but also all possible mechanisms that ensure $\varepsilon$-local differential privacy. As a result, classical information processing inequalities, such as Fano's inequality, fail to capture the interplay between privacy constraints and statistical accuracy.

To address this, we employ stronger information processing inequalities adapted from \citet{6686179}, enabling us to provide lower bounds over both the class of estimators and the mechanisms that ensure $\varepsilon$-local differential privacy. Moreover, we demonstrate that the stated lower bound can be achieved using a simple privatization mechanism coupled with a spectral estimator. 

It is worth emphasizing that enforcing local differential privacy at a user-specified level substantially alters the degree heterogeneity in the original network. Consequently, the dependence of the minimax rate on the degree distribution departs significantly from the non-private setup studied in \cite{JIN2023,ke2022optimal}. In particular, the minimax-optimal algorithm proposed in \cite{ke2022optimal} stops to be optimal in this regime, owing to its reliance on pre-PCA normalization of node embeddings based on the graph Laplacian. When the degrees are perturbed by edge-level privacy noise, such normalization introduces additional variance, leading to suboptimal estimation. To address this issue, we employ a modified version of the \textsc{MIXED-SCORE} algorithm of \cite{JIN2023}, incorporating privacy-aware adjustments. Through a careful spectral perturbation analysis based on the leave-one-out technique of \cite{ke2022optimal}, we establish that the modified \textsc{Mixed-Score} estimator attains the minimax-optimal rate for estimating the membership profile matrix~$\Pi$ under user-specified local differential privacy constraints.

\subsection{Paper Organization} 
\textcolor{black}{In Section~\ref{sec:price}, we establish the information-theoretic price of privacy for estimating $\Pi$ from a random network generated under the model in \eqref{eq:DCMM}, while adhering to edge-level local differential privacy constraints. Specifically, we derive minimax lower bounds for the estimation error of $\Pi$ under the loss function $\mathcal{L}$ defined in \eqref{eq:loss}. In Section~\ref{prime}, we present a mechanism for releasing data that ensures local differential privacy. Additionally, we propose a spectral algorithm that operates on the privatized data matrix to estimate $\Pi$, achieving the optimal accuracy characterized in Section~\ref{sec:price}.} 

\textcolor{black}{In Section~\ref{numerical}, we support our theoretical findings with simulation studies. Finally, in Section~\ref{real_data}, we demonstrate the effectiveness of our proposed algorithm on two real-world datasets: the Facebook Ego Network and the Political Blog Dataset.
}

\subsection{Notations and Assumptions}
In order to facilitate our discussion throughout this paper, we introduce some notation. Firstly, we shall denote the set $\{1,\ldots,m\}$ as $[m]$. Secondly, for any given matrix $A$, we will denote its spectral norm as $\|A\|$ and the Frobenius norm as $\|A\|_F$. The $i$-th row and $j$-th column of matrix $A$ will be denoted by $A_{i*}$ and $A_{*j}$ respectively. Furthermore, the matrix norm $\|A\|_{2 \rightarrow \infty}:=\max_{1 \le i \le n}\|A_{i*}\|_2$. For any vector $v,$ its $L^p$ norm will be denoted by $\|v\|_p$ for $p\ge 1.$ For any two sequences $a_n$ and $b_n,$ $a_n\asymp b_n$ means there exists constants $C_1,\,C_2>0$ such that $a_n\le C_1b_n$ and $b_n\le C_2a_n$ for all sufficiently large $n,$ and $a_n\lesssim b_n$ means there exists constant $C>0$ such that $a_n\le Cb_n$ for all sufficiently large $n.$ The Kullback-Leibler divergence between two probability distributions $P$ and $Q$ will be denoted by $\KL(P\|Q)$ and the total variation distance between them shall be denoted by $\|P-Q\|_{\mathrm{TV}}$. The vector of all ones in $K$ dimension will be denoted by $1_K$. Next, the adjacency matrix of a graph $\mathscr G$ will be denoted by $A_{\mathscr G}$ and we shall assume that $(A_{\mathscr G})_{ii}=0$ for all $i \in [n]$.
%Finally, we assume that the diagonal elements of $\Theta$, namely, $\theta_1,\ldots,\theta_n$, are uniformly bounded.

%\paragraph{Related Work:} 

\section{Balancing Privacy and Statistical Utility}
\label{sec:price}
\subsection{Local Differential Privacy and Problem Statement}
In the context of privacy-preserving analysis of network data, two primary notions of local privacy are widely studied: \emph{node Local Differential Privacy} (node LDP) \cite{kasiviswanathan2013analyzing} and \emph{edge Local Differential Privacy} (edge LDP) \cite{nissim2007smooth}. While both aim to protect individuals' privacy in graph-structured data, they operate at different levels of granularity and offer distinct privacy guarantees. In this paper, we focus primarily on edge Local Differential Privacy (edge LDP) as formalized in \cite{imola2021locally}.

Let $\mathcal{A}$ denote the space of adjacency matrices on $n$ vertices. Consider randomized functions $\mathcal{M}_{ij}: \{0,1\} \to \mathbb{R}$ acting on the $(i,j)$-th entry of $A_{\mathscr{G}}$, and define the function $\mathcal{M}:\mathcal{A} \to \mathbb{R}^{n \times n}$ entrywise as $(\mathcal{M}(A_{\mathscr{G}}))_{ij} = \mathcal{M}_{ij}((A_{\mathscr{G}})_{ij})$ for $(i,j) \in [n] \times [n]$. 

\begin{defn}[$\varepsilon$-edge LDP]
\label{def:edge_ldp}
A mechanism $\mathcal{M}$ satisfies $\varepsilon$-edge Local Differential Privacy if, for any $A_{\mathscr{G}}, A_{\mathscr{G}'} \in \mathcal{A}$ such that ${\mathscr{G}}$ and ${\mathscr{G}'}$ differ by exactly one edge, we have:
\[
\frac{\P(\mathcal{M}(A_{\mathscr{G}}) \mid A_{\mathscr{G}})}{\P(\mathcal{M}(A_{\mathscr{G}'}) \mid A_{\mathscr{G}'})} \leq e^\varepsilon.
\]
\end{defn}

\textcolor{black}{
In this paper, we aim to characterize the information-theoretic bottlenecks in designing mechanisms $\mathcal{M}$ that ensure local $\varepsilon$-edge differential privacy while enabling accurate estimation of latent membership profiles at later stages of data analysis.
}

\textcolor{black}{
Our primary focus lies in the design of privatizing protocols. Specifically, we assume access to the raw data $\mathscr{G}$ and are tasked with releasing a function of this data, $\mathcal{M}(A_{\mathscr{G}})$, by adding appropriate noise to ensure that the released dataset satisfies the conditions of Definition \ref{def:edge_ldp}. Furthermore, the mechanism $\mathcal{M}$ must be designed such that natural estimators can be employed by a statistician at a later stage to infer the unknown membership profile matrix $\Pi$.
}

\textcolor{black}{
A related but distinct problem of significant statistical interest involves analyzing network data with adversarial perturbations. In this scenario, the problem is approached from the perspective of a statistician who receives a dataset modified to ensure privacy, without knowledge of the exact privatizing protocol. The analyst must then attempt to estimate the matrix $\Pi$. As demonstrated in \citet{9996945}, consistent estimation of $\Pi$ may not be feasible in the absence of knowledge about $\varepsilon$ and the mechanism used to alter the raw data. While the problem of characterizing minimax lower bounds for the risk of estimating $\Pi$ from adversarially perturbed data when the true data is generated from a mixed membership SBM is undoubtedly important, it lies beyond the scope of this work. We believe this is an interesting direction for future research but emphasize that our focus here is specifically on designing privatizing mechanisms to address the tradeoff between privacy and statistical utility.
}

\subsection{Minimax Lower Bound to the Estimation Risk under Privacy Constraint}
Understanding fundamental trade-offs between privacy constraints and statistical accuracy of any procedure to estimate $\Pi$ from a privatized dataset requires the characterization of the minimax lower bound to the estimation risk $\mathcal L(\widehat \Pi,\Pi)$. In this section, we derive a lower bound on the estimation risk for inferring the membership profile matrix $\Pi$ under $\varepsilon$-edge local differential privacy (LDP). 

In that direction, let us consider some additional assumptions on the data-generating model \eqref{eq:DCMM}. Let us assume that $\varepsilon \in [0,\varepsilon_0]$ and 
\begin{align}
    \label{eq:bounded_theta}
    \max_{i \in [n]}\theta_i \le C,
\end{align}
for some $C>0$.
%Let $d_i=\sum_{j=1}^nA_{ij}$ be the degree of node $i$ and $\bar d_n=n^{-1}\sum_{i=1}^nd_i$ be the average degree. We define $\bar \Omega:=\Theta\Pi B\Pi^\top\Theta-\mathsf{diag}(\Theta\Pi B\Pi^\top\Theta)$ and $D_\theta$ to be the diagonal matrix $D_\theta := \mathsf{diag}(\bar D_{11},\ldots,\bar D_{nn})$ where $\bar D_{ii}=\mathbb E(d_i+\bar d_n)$ for $i \in [n]$. Next, we define 
Next, consider the matrix
$$
G = K \|\theta\|^{-2}_2(\Pi^\top\Theta^2\Pi).
$$
Assume that $B_{ii}=1$ for all $i \in [n]$, and
\begin{align}
\label{eq:singularity_b}
 &\min_{i\neq j}B_{ij} \ge c, \quad  \|G\| \le c_1,\quad\mbox{and} \quad \|G^{-1}\| \le c_1   
\end{align}
for some constants $c,\,c_1>0.$ These assumptions ensures that the degree distributions across different communities do not vary significantly. In the special case when all the nodes are pure and there is no degree of heterogeneity, this condition translates to the uniform distribution of community sizes.

Next, let us consider the ordered eigenvalues of $BG$ denoted by $\lambda_1(BG),\ldots,\lambda_K(BG)$ in decreasing order. We assume that there exists $c_2 \in (0,1)$ and $\beta_n \in (0,1)$ such that
\begin{align}
\label{eq:spectral_gap}
   & |\lambda_2(BG)| \le (1-c_2)\lambda_1(BG),\quad\mbox{and}\quad\max_{k\in \{2,\ldots,K\}}\lambda_k(BG) \asymp \beta_n.
\end{align}
The first inequality ensures that the leading eigenvector of the true signal matrix $\Pi^\top B\Pi$ is separated from the remainder of spectrum which in turn helps in correcting the estimator for degree heterogeneity. The second inequality in \eqref{eq:spectral_gap} is a spectral gap assumption common in the literature of spectral clustering. The term $\beta_n$ represents the signal strength which governs the rate of convergence of the spectral estimator.

There exists a constant $c_3>0$ such that
\begin{align}
\label{eq:perron_root}
\min_{1\le k\le K}\left\{\sum_{i=1}^n\theta^2_i\Pi_{ik}\right\} \ge c_3\max_{1\le k\le K}\left\{\sum_{i=1}^n\theta^2_i\Pi_{ik}\right\}.
\end{align}
If the eigenvectors of $BG$ are denoted by $\eta_1,\ldots,\eta_K$, by Proposition A.1 of \cite{JIN2023}, we can conclude using \eqref{eq:singularity_b} and \eqref{eq:perron_root} that there exists constants $c_4>0$ such that
\begin{equation}
\min_{1 \le j \le K}\,\eta_{1j}>0, \quad \mbox{and} \quad \frac{\max_{1 \le j \le K}\eta_{1j}}{\min_{1 \le j \le K}\eta_{1j}}\le c_4.
\label{eq:BG_eigenvector}
\end{equation}
%\textcolor{orange}{Comment: It would be beneficial to discuss why these assumptions are standard and where they appear in the non-private literature. Additionally, discuss any new assumptions we need that are not present for the non-private algorithms.}
We shall also assume that there exists a constant $c_5>0$ such that,
\begin{equation}
\label{eq:non_negligible_presence}
\{1 \le i \le n:\Pi_{ik}=1, \theta_i\ge c_5\widetilde \theta\}\neq \emptyset 
\end{equation}
for all $1 \le k \le K$, where the degree correction factor $\widetilde\theta$ is defined as
\begin{align}
\label{eq:def_tilde_theta_main_param}
\widetilde\theta=\|\theta\|_2/\sqrt{n}.
\end{align}
This condition ensures that each community contains at least one pure node with a significantly large degree which is sufficent for the model \eqref{eq:DCMM} to be identifiable, as discussed in Remark \ref{rem1_1}. {Furthermore, Proposition A.2 of \cite{JIN2023} provides the conditions on $\Pi$, $B$, and $\Theta$ that ensure property \eqref{eq:BG_eigenvector}.} The assumptions stated in \eqref{eq:bounded_theta}-\eqref{eq:perron_root} and \eqref{eq:non_negligible_presence} are standard regularity assumptions adopted in the non-private literature \cite{JIN2023, ke2022optimal}.

The lower bound to the estimation risk will be governed by the following parameter:
\begin{equation}
    \label{eq:err_n_old}
    \mathrm{err}_n:=\frac{K^{3/2}}{\beta_n\sqrt{n}\,\widetilde \theta^2}\left(\frac{e^\varepsilon +1}{e^\varepsilon - 1}\right), \quad \mbox{where $\widetilde \theta$ is defined in \eqref{eq:def_tilde_theta_main_param}.}
\end{equation}
% \textcolor{brown}{
% \begin{equation}
%     \mathrm{err}_n:=\frac{K^{3/2}}{\beta_n\sqrt{n\bar\theta^4}}\left(\frac{e^\varepsilon +1}{e^\varepsilon - 1}\right).
%     \label{eq:err_n_old}
% \end{equation}
% }
Let $F_n$ be the empirical distribution function of $\{\theta_1/\widetilde\theta,\ldots,\theta_n/\widetilde\theta\}$. Let us consider the specific collection of degree heterogeneity parameters $\mathcal{G}_n(\gamma,\,a_0) \subset \R^n$ defined as follows:
%\nb{Sayak figure out what condition is needed.}
\begin{defn}
\label{def:constr_deg_cor}
For constants $\rho\in(0,\,1)$ and $a_0\in(0,\,1),$ let $\mathcal{G}_n(\rho,\,a_0)$ be the collection of all $\theta\in\R^n$ such that there exists $c_n>0$ satisfying $\mathrm{err}_n \le\rho c_n$, and
$$\int_{\rho c_n}^{c_n}\frac{\text{d}F_n(t)}{t\wedge 1}\ge a_0\int_{\mathrm{err}_n}^\infty\frac{\text{d}F_n(t)}{t\wedge 1},$$
where $\mathrm{err}_n$ is defined by \eqref{eq:err_n_old}.
\end{defn}
% The above definition is satisfied by when $\theta_1,\ldots,\theta_n$ are generated independently and identically as $\theta_i=\kappa_n\omega_i$ where $\omega \overset{\mathrm{i.i.d}}{\sim} F$, where $\kappa_n$ is an appropriate constant and any one of the following properties hold:
% \begin{enumerate}
%    \item $F=\sum_{i=1}^L\epsilon_i\delta_{x_i}$ with $\min_k \epsilon_k>0$ and $x_1>0$. 
%     \item $F(\cdot)$ is a continuous distribution with support in $[c,\infty)$ and finite variance.
%  \end{enumerate}

\begin{rem}
    While this condition appears to be technical, we can modify the arguments of Lemma A.2 of \cite{ke2022optimal} slightly to show that these conditions hold for the following classes of degree distribution.
    \begin{enumerate}
    \item A finite mixture of point masses $F = \sum_{k=1}^L \epsilon_k \delta_{x_k}$, where $\min_k \epsilon_k > 0$ and $x_k > 0$ for all $k$.
    \item A continuous distribution supported on $[c,\infty)$ for some constant $c>0$.
\end{enumerate}
\end{rem}

% provide a few examples of degree distributions satisfying this constraint.

% \begin{prop}
% \label{prop:prop_generality}
% Assume that $\theta_1,\ldots,\theta_n$ are generated independently from a distribution $F$ on $(0,\infty)$ with finite variance. Furthermore, suppose that $\mathrm{err}_n \to 0$ as $n \to \infty$. Then there exists an $n_0 \ge 0$ such that, for all $n \ge n_0$, the constraint in Definition~\ref{def:constr_deg_cor} is satisfied with probability at least $1 - o(1)$ for each of the following classes of distributions $F$:
% \begin{enumerate}
%     \item A finite mixture of point masses $F = \sum_{k=1}^L \epsilon_k \delta_{x_k}$, where $\min_k \epsilon_k > 0$ and $x_k > 0$ for all $k$.
%     \item A continuous distribution supported on $[c,\infty)$ for some constant $c>0$.
% \end{enumerate}
% \end{prop}

%For the proof of the above proposition, please refer Section SM2. 
{To construct the lower bound, let us consider $\mathfrak M_\varepsilon$, the class of all estimators of $\Pi$ that have been constructed from a dataset released by the provider after appropriately altering the observed adjacency matrix using any randomized mechanism to enforce $\varepsilon$-edge local differential privacy as defined in Definition \ref{def:edge_ldp}. The following result provides a lower bound to the minimax risk given by:
\begin{equation}  
\label{eq:def_ldp_min_max_risk}
\inf_{\widehat\Pi \in \mathfrak M_\varepsilon}\sup_{(\Pi,\,B)\in\mathcal{Q}_n(\theta)}\E\mathcal{L}(\widehat\Pi,\,\Pi),
\end{equation}
where for any sequence of degree correction parameters $\theta \in \mathcal{G}_n(\rho,\,a_0)$, the collection $\mathcal{Q}_n(\theta)$ contains the matrices $\Pi \in \R^{n \times K}$ and $B \in \R^{K \times K}$ satisfying the assumptions \eqref{eq:singularity_b}, \eqref{eq:spectral_gap}, \eqref{eq:perron_root}, and \eqref{eq:non_negligible_presence}.
\begin{rem}
Observe that, a lower bound on the risk defined by \eqref{eq:def_ldp_min_max_risk} characterizes the unavoidable error incurred in estimating $\Pi$, accounting for both the choice of the estimator and the privatizing mechanism employed to ensure local differential privacy.
\end{rem}}
\begin{figure}
    \centering
    \includegraphics[width=0.75\linewidth]{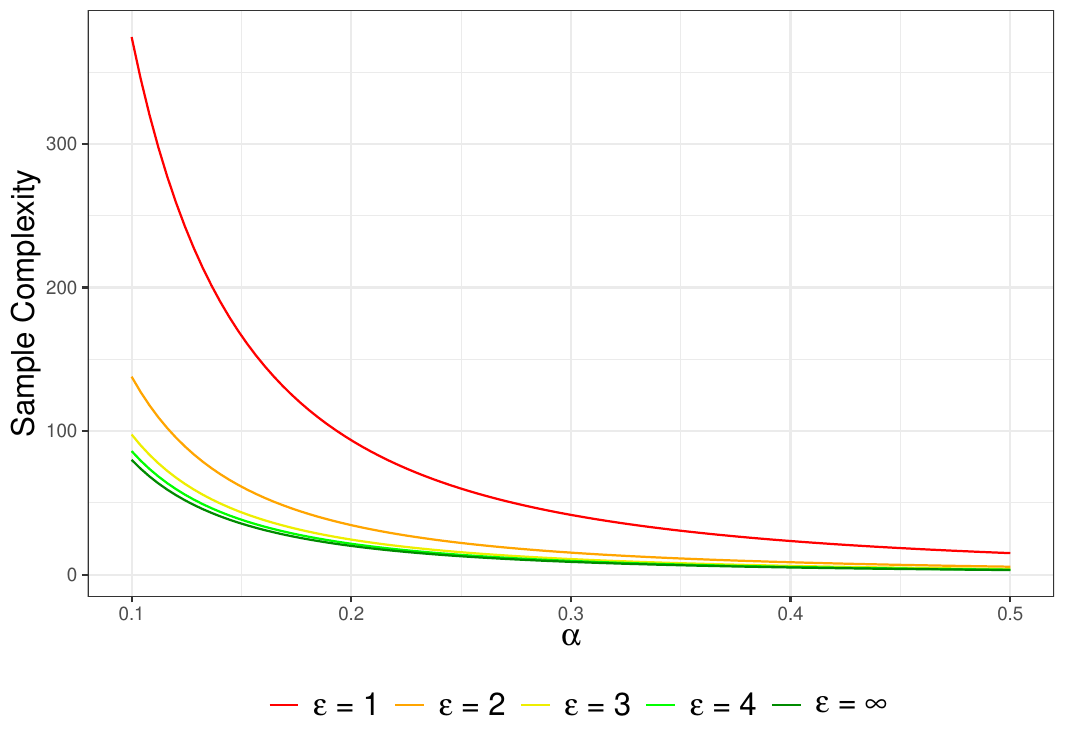}
    \caption{Sample Complexity as a function of estimation error $\alpha$ and privacy budget $\varepsilon$}
    \label{fig:samp_complex_v_eps}
\end{figure}
\begin{thm}
{Given constants $\rho\in(0,1)$ and $a_0\in(0,1),$ consider $\theta \in \mathcal G_n(\rho,a_0)$ such that $\mathrm{err}_n\to0$. Furthermore, consider an $\varepsilon \in [0,\varepsilon_0]$ for some $\varepsilon_0>0$ and $\theta_i\in[0,\,C]$ for all $1\le i\le n$ for some constant $C>0$. Then there exists a constant $C^\prime_0>0$ (independent of $\theta$, $n$ and $\varepsilon$) such that for all sufficiently large $n$, we have
\begin{align*}
\inf_{\widehat\Pi \in \mathfrak M_\varepsilon}\sup_{(\Pi,\,B)\in\mathcal{Q}_n(\theta)}\E\mathcal{L}(\widehat\Pi,\,\Pi)&\ge \frac{C^\prime_0}{n}\sum_{i=1}^n\min\left\{\frac{\mathrm{err}_n}{(\theta_i/\widetilde \theta)\wedge1},1\right\}=C_0'\int_0^\infty\min\left\{\frac{\mathrm{err}_n}{t\wedge1},\,1\right\}\text{d}F_n(t),
\end{align*}
where $F_n$ is the empirical distribution function of $\{\theta_1/\widetilde\theta,\,\theta_2/\widetilde\theta,\,\ldots,\,\theta_n/\widetilde\theta\}$  with $\widetilde \theta=\|\theta\|_2/\sqrt{n}$.
    \label{thm:lower_bound_informal}}
\end{thm}
\begin{rem}
{
We require $\mathrm{err}_n \to 0$ for Theorem~\ref{thm:lower_bound_informal} to hold, which imposes the following condition on the signal strength~$\beta_n$:
\[
\beta_n \gg \frac{1}{\sqrt{n}\,\widetilde{\theta}^2\,\tanh(\varepsilon/2)}.
\]
Compared to the analogous condition on~$\beta_n$ in Theorem~4.2 of~\citet{ke2022optimal}, enforcing differential privacy demands a stronger separation between the within-community and inter-community connection probabilities, particularly for small values of $\widetilde \theta$ which is an characteristic of a sparse network. This tightening is necessitated by the additional randomness in the degree distribution introduced through edge-level privatization. Indeed, any local edge-DP randomization amplifies uncertainty in node degrees, thereby making the estimation of the underlying degree heterogeneity substantially more difficult and worsening the dependence of the estimation rate for any estimator of~$\Pi$ on the baseline degree in the entire network. This deviation from the non-private regime is particularly pronounced for sparse networks.}
\end{rem}
{If the degree heterogeneity parameters $\theta_i$ are approximately equal to $\widetilde{\theta}$ for all $i \in [n]$, then the above theorem guarantees that to achieve an error of $\alpha \in (0,1)$, the minimum number of samples required is lower bounded by:
\[
\widetilde{\mathcal{O}}\left(\frac{K^{3}}{\beta_{n}^{2} \widetilde{\theta}^{4} \alpha^{2}} \tanh^{-2}\left(\frac{\varepsilon}{2}\right)\right),
\]
up to constants, where $\widetilde{\mathcal{O}}(\cdot)$ also suppresses logarithmic factors.. For smaller values of $\varepsilon$, this simplifies to:
\[
\widetilde{\mathcal{O}}\left(\frac{K^{3}}{\beta_{n}^{2} \widetilde{\theta}^{4} \alpha^{2} \varepsilon^2}\right),
\]
}
{Compared to the analogous lower bound in Theorem~2.1 of~\citet{ke2022optimal}, for small values of~$\varepsilon$, the effective sample size reduces to~$n\varepsilon^2$, which can be interpreted as the \emph{price of privacy}. Moreover, when the degree heterogeneity parameter~$\widetilde{\theta}$ is small, the private setting requires $\mathcal O(\widetilde{\theta}^{-4})$ samples in contrast to $\mathcal O(\widetilde{\theta}^{-2})$ in the non-private case to achieve the same estimation accuracy for~$\Pi$ in the worst case. The lower bound in Theorem \ref{thm:lower_bound_informal} demonstrates that this deteriorated dependence on $\widetilde{\theta}$, relative to the non-private setting, is intrinsic. It arises from the stringent requirement of enforcing privacy at the edge level, which inevitably suppresses information about the degree heterogeneity of the underlying network in the privatized data—regardless of the specific mechanism used to ensure privacy. The dependence of the sample complexity on the estimation error~$\alpha$ is illustrated in Figure~\ref{fig:samp_complex_v_eps}. As observed, the sample complexity increases sharply for small~$\varepsilon$, but decreases rapidly as~$\varepsilon$ grows. 
}
%\nb{Need to provide a graph showing the dependence on theta}
%\textcolor{orange}{In the remark: Do we need to stress on the fact that this modified dependence on degree heterogeneity is not an artifact of our privacy mechanism but something fundamental, as is implied by our lower bounds}

{Let us define $\nu_i:=\theta_i/\widetilde\theta$ where $\widetilde \theta=\|\theta\|_2/\sqrt{n}$. 
The claimed lower bound in Theorem \ref{thm:lower_bound_informal} can be decomposed as
\begin{equation}
    %\int_0^\infty\min\left\{\frac{\mathrm{err}_n}{t\wedge1},\,1\right\}\text{d}F_n(t)=
    \frac{1}{n}\sum_{i=1}^n\min\left\{\frac{\mathrm{err}_n}{\nu_i\wedge1},1\right\}=\frac{|\{i:\nu_i\le\mathrm{err_n\}}|}{n}+\frac{1}{n}\sum_{i:\nu_i>\mathrm{err}_n}\frac{\mathrm{err}_n}{\nu_i\wedge 1}
    \label{eq:lb_decomp}
\end{equation}
since we assume $\mathrm{err}_n\ll1.$ To prove the lower bound on the estimation rate of the membership profile matrix $\Pi$ under an arbitrary data privatization scheme that satisfies $\varepsilon$ local edge differential privacy as defined in Definition \ref{def:edge_ldp}, we shall consider two separate cases depending on which of the term in \eqref{eq:lb_decomp} dominates the other. In that direction, as the first regime, we shall consider the degree heterogeneity parameters $\theta_1,\ldots,\theta_n$ such that
\begin{align}
\frac{\left|\{i:\nu_i \le \mathrm{err}_n\}\right|}{n} \le \frac{C}{n}\sum_{i:~\nu_i> \mathrm{err}_n}\frac{\mathrm{err}_n}{\nu_i \wedge 1},
    \label{eq:F_n-condn}
\end{align}
where $\mathrm{err}_n$ is defined in \eqref{eq:err_n_old} and $C>0$ is a $n$ independent constant.}

% To prove the above theorem, we consider two separate cases. For the first case, we assume that
% \textcolor{purple}{\begin{equation}
% \label{eq:sharp_cons_v}
%     \frac{\left|\{i:\nu_i \le \mathrm{err}_n\}\right|}{n} \le \frac{C}{n}\sum_{i:~\nu_i \ge \mathrm{err}_n}\frac{\mathrm{err}_n}{\nu_i \wedge 1},
% \end{equation}
% where $\nu_i = \theta_i/\widetilde{\theta}$ for all $i \in [n]$.}
% \nb{Should we mention why this particular condition? See in lower bound proof.}
This condition refers to the situation when the proportion of nodes with relative average degree $\nu_i=\theta_i/\widetilde{\theta}$  less than $\mathrm{err}_n$ has a negligible impact in the loss $\mathcal L$. Therefore, for such nodes, we can afford to make some errors in estimating such membership profiles. Furthermore, for the nodes with relatively large $\theta_i/\widetilde{\theta}$, the estimation of the membership profiles will be inherently accurate due to higher degrees. This leaves us with a set of $n_0$ nodes with moderate relative average degrees for which estimating the membership profiles is the most challenging. 

{To characterize the lower bound for the minimax risk defined in \eqref{eq:def_ldp_min_max_risk}, we transform the estimation problem into a testing problem where nature selects a collection of parameters and an arbitrary mechanism to privatize the data. The user observes the privatized data and must infer the parameter configuration that generated the true dataset. }

{The core idea of the proof is to consider $n_0 \asymp n$ and construct a collection of $2^{\Theta(n_0K)}$ configurations of $\Pi$ by perturbing the membership profiles of these $n_0$ nodes. These configurations are designed to be well-separated in terms of the loss $\mathcal{L}$, while ensuring that the distance between the distributions of privatized data generated using the stated configurations, after applying the randomized mechanism, is small. Since the privatizing mechanism is unknown to the user, we need to demonstrate that this distance is uniformly small over all possible randomization mechanisms ensuring local differential privacy. Specifically, let $\mathcal{R}_\varepsilon$ denote the collection of randomized mechanisms that satisfy $\varepsilon$-edge local differential privacy for adjacency matrices. We then need to show that:
\begin{equation}
\label{eq:private_kl_1}
\sup_{\mathcal{M} \in \mathcal{R}_\varepsilon} \KL(\mathcal{M}(A_{\Pi^{(j)}})\| \mathcal{M}(A_{\Pi^{(0)}}))\le\alpha\log J,
\end{equation}
for some $\alpha \in (0,1/8)$ where $\Pi^{(0)}, \ldots, \Pi^{(J)}$ are the configurations selected by nature. The rest of the theorem follows using an application of Fano's inequality and Pinsker's inequality.}

{In contrast to the techniques used in \citet{ke2022optimal}, which rely on bounding:
\[
\KL(A_{\Pi^{(j)}} \| A_{\Pi^{(0)}}),
\]
proving an upper bound for \eqref{eq:private_kl_1} requires the use of stronger data processing inequalities adapted from \citet{6686179}. These inequalities account for the added complexity of the privatization mechanism, which introduces additional randomness to the data. The application of such techniques to establish minimax lower bounds for recovering latent communities under the constraint of local differential privacy has not been explored in the past, marking a key novelty of our work.}

When the condition \eqref{eq:F_n-condn} fails to hold, the error in estimating the membership profiles for nodes with a relative average degree below $\mathrm{err}_n$ significantly affects the total estimation error. In this case, the least favorable configurations are constructed by perturbing the membership profiles of all such nodes. The remainder of the proof follows a similar strategy to the previous case, with detailed steps provided in the appendix.

% \begin{rem}
% \nb{Sayak to address this remark}
% \textcolor{black}{It is important to note that our lower bound holds under the sequence of degree correction parameters satisfying the conditions stated in Definition \ref{def:constr_deg_cor}. While these conditions may appear technical, they are broad and encompass many natural sequences of $\theta$. \textcolor{purple}{A detailed discussion on the generality of the class $\mathcal{G}_n(\rho, a_0)$ is provided in Section A.2 of \citet{ke2022optimal}.}}
% \end{rem}

\section{PriME: Optimal Mechanism for privatizing the data and Membership Profile Estimation}
\label{prime}
\subsection{Edge Flip Mechanism to enforce Local Differential Privacy}
In this section, we describe a specific randomization scheme $\mathcal{M}$ designed to enforce edge local differential privacy in the resulting network. The core idea of this approach, commonly referred to as the \emph{randomized response mechanism}, involves generating a synthetic copy of the original network. This is achieved by constructing an adjacency matrix $\mathcal{M}_\varepsilon(A_{\mathscr{G}})$, where each edge in the original graph ${\mathscr{G}}$ is randomly flipped with a specified probability to ensure privacy at the user-specified level.
 Using the property of differential privacy, any analysis performed on the synthetic network $\mathcal{M}_\varepsilon(A_{\mathscr G})$ preserves $\varepsilon$–edge LDP. This edge-flipping approach has been utilized in previous studies. It has been interpreted under various definitions, such as (central) edge DP, edge LDP, and relationship DP. In particular, we use \emph{Symmetric Edge-Flip Mechanism} proposed in \cite{imola2021locally}, to define $\mathcal M$ in our setup.
\begin{defn}[Symmetric Edge-Flip Mechanism]
Suppose that $\varepsilon>0$. For $i \in[n]$, let $\mathcal{M}_i:\{0,1\}^n \rightarrow\{0,1\}^n$ be such that
\begin{equation}
\label{eq:sym_edge_flip}
 \left[\mathcal{M}_i(x)\right]_j \stackrel{\text { ind }}{=}
\begin{cases}
0 & \text{if } i \geq j \\
1-x_j & \text{if } i<j \text{ w.p. } \frac{1}{1+e^{\varepsilon}} \\
x_j & \text{if } i<j \text{ w.p. } \frac{e^\varepsilon}{1+e^{\varepsilon}}
\end{cases}   
\end{equation}
and let
$$
T(A_{\mathscr G})=\left[\begin{array}{c}
\mathcal{M}_1\left((A_{\mathscr G})_{1 *}\right) \\
\vdots \\
\mathcal{M}_n\left((A_{\mathscr G})_{n *}\right)
\end{array}\right].
$$
Then the $n \times n$ symmetric edge-flip mechanism is the mechanism 
\[\mathcal{M}_{\varepsilon}(A_{\mathscr G})=T(A_{\mathscr G})+[T(A_{\mathscr G})]^\top.\]
\end{defn}
\begin{thm}
    The symmetric edge-flip mechanism $\mathcal{M}_{\varepsilon}$ satisfies $\varepsilon$-edge LDP.
\end{thm}
The proof of this theorem follows from Theorem 3.6 of \cite{Hehir2021ConsistencyOP}. 

{One approach to release data while ensuring edge local differential privacy is to provide the matrix $\mathcal{M}_{\varepsilon}(A_{\mathscr{G}})$ to the practitioner. However, note that:
\[
\E[\mathcal{M}_\varepsilon(A_{\mathscr{G}})] = p_\varepsilon(1_n1^\top_n - I_n) + (1 - 2p_\varepsilon)(\Omega - \mathsf{diag}(\Omega)),
\]
where $\Omega \equiv \Theta \Pi B \Pi^\top \Theta$, $\Theta \in \mathbb{R}^{n \times n}$ is given by $\Theta = \mathsf{diag}(\theta_1, \ldots, \theta_n)$, and $p_\varepsilon = (1 + e^\varepsilon)^{-1}$. The randomized response mechanism introduces a high-rank bias $p_\varepsilon(1_n1^\top_n - I_n)$ in $\mathcal{M}_\varepsilon(A_{\mathscr{G}})$. For very small values of $\varepsilon$, this bias can completely obfuscate the low-rank signal $\Omega - \mathsf{diag}(\Omega)$, rendering consistent estimation of $\Pi$ infeasible with natural estimators applied to $\mathcal{M}_\varepsilon(A_{\mathscr{G}})$.}

{To address this, we propose de-biasing and rescaling $\mathcal{M}_\varepsilon(A_{\mathscr{G}})$ before releasing the data to the statistician. Specifically, we recommend releasing the centered and scaled matrix $\pM$, defined as:
\begin{align}
\label{eq:def_pm}
\pM = \frac{1}{1 - 2p_\varepsilon}\bigg(\mathcal{M}_\varepsilon(A_{\mathscr{G}}) - p_\varepsilon(1_n1^\top_n - I_n)\bigg).
\end{align}
By the composition properties of differential privacy, the matrix $\pM$ also satisfies $\varepsilon$-edge local differential privacy. Importantly, the transformed matrix satisfies:
\[
\E[\pM] \equiv \Omega - \mathsf{diag}(\Omega),
\]
allowing the application of natural spectral algorithms for estimating membership profiles in the mixed membership SBM. This transformation obfuscates individual edge information to protect the privacy of connections while preserving the essential low-rank signal $\Pi B \Pi^\top$ for the analyst.}

{However, it is also worth noting that such randomization substantially alters the degree distribution of the nodes in the underlying graph. In particular, the variance of the average degree of node~$i$ increases from~$\mathcal O(\theta_i)$ to~$\mathcal O(1/\varepsilon^2)$, making the estimation of~$\theta_i$ considerably more challenging. This, in turn, affects the performance of any estimator that relies on the estimated~$\theta_i$ to correct for degree heterogeneity. Nevertheless, as discussed in the previous section, this degradation in performance for estimating the degree parameters~$\{\theta_i:i \in [n]\}$ is intrinsic to the privacy-constrained setting and not an artifact of our proposed procedure. In fact, we show that there exists a simple spectral algorithm capable of leveraging the privatized adjacency matrix to achieve the minimax-optimal rate for estimating the membership profile matrix~$\Pi$. This analysis highlights that the decision-theoretic behavior of the private membership estimation problem is fundamentally distinct from its non-private counterpart.
}

\begin{rem}
{It is worth noting that the privatized data matrix $\pM$ released to the analyst after the randomization scheme is no longer an adjacency matrix of a network. A natural question that arises is whether it is possible to transform $\pM$ into an adjacency matrix of another network without compromising the signal in the data. At present, we do not have a satisfactory solution to this question, and we leave it as an open problem for future research.}
\end{rem}

\subsection{Estimation of the Community Membership Profiles}
\label{def_alg_terms}
Let us consider the centered and scaled matrix $\pM$ defined in \eqref{eq:def_pm}. In this section, we shall provide a spectral algorithm that can be used by a statistician to estimate $\Pi$ from $\pM$ with the optimal accuracy stated in Theorem \ref{thm:lower_bound_informal}.

{In this direction, we employ the post-PCA SCORE normalization~\cite{article_score} to estimate~$\Pi$ from the privatized matrix~$\pM$. The post-PCA normalization mitigates the effect of high degree heterogeneity and enables optimal recovery of~$\Pi$. In the non-private setting, however, an additional pre-PCA normalization is typically applied to correct for extremely low-degree nodes~\cite{ke2022optimal}. Under the local-DP model, the increased variance of node degrees induced by edge-level randomization renders such pre-PCA normalization ineffective, and, in fact, detrimental to the estimation accuracy. Consequently, our algorithm operates directly on the privatized matrix~$\pM$ without any pre-PCA adjustment.
}

\iffalse
, it is crucial to mitigate the impact of degree heterogeneity among the nodes. This can be achieved through \emph{pre-} and \emph{post-Principal Component Analysis} (PCA) normalization. In that direction, taking inspiration from \cite{ke2022optimal}, we introduce the modified matrix $L_0$ constructed as follows:
\begin{enumerate}
    \item First, we define the pseudo-degrees  $d_1,\ldots,d_n$ of the nodes as follows: 
    $$
d_i = \sum_{j \neq i}(\pM)_{ij}.
$$
\item Next, we calculate the average modified degree $\bar d_n$ where $\bar d_n=\frac{1}{n}\sum_{i=1}^{n}d_i$.
\item Finally, we introduce the matrix $H=\mathsf{diag}(d_1,\ldots,d_n)+\tau\bar d_nI_n$ and define the modified Graph Laplacian as follows:
\begin{align}
\label{eq:def_graph_laplacian}
  \pL = H^{-1/2}\pM H^{-1/2}.  
\end{align}
\end{enumerate}
Now, let us consider the population version of $\pL$ given by
\begin{align}
\label{eq:popln_laplacian}
L_0 = H^{-1/2}_0\Theta\Pi B\Pi^\top\Theta H^{-1/2}_0 
\end{align}
where $H_0=\mathbb E[\mathsf{diag}(d_1,\ldots,d_n)+\tau\bar d_nI_n]$. Here $\tau>0$ represents a regularization parameter specified by the user which helps to ensure the invertibility of both the matrices $H$ and $H_0$.
\fi

To motivate our algorithm, consider the low-rank matrix $\Omega=\Theta\Pi B\Pi^\top\Theta$. Let $\Xi \in \R^{n \times K}$ denote the matrix of its eigenvectors and $\lambda_1,\ldots,\lambda_K$ denote its eigenvalues. Then by definition, there exists a non-singular matrix $V \in \R^{K \times K}$ such that $\Xi=\Theta\Pi V$. We apply the post-PCA normalization developed in \citet{article_score} to $\Xi$ and construct the matrix $R \in \R^{n \times (K-1)}$ where $R_{ij} = (\Xi)_{i(j+1)}/(\Xi)_{i1}$, for $i \in \{1,\ldots,n\}$ and $j \in \{1,\ldots,K-1\}$. Using Lemma 2.1 of \citet{ke2022optimal}, the rows of $\widetilde{V} \in \R^{(K-1)\times K}$, given by $\widetilde{V}_{ij} = V_{(i+1)j}/V_{1j}$
represent the vertices of a simplex and all the rows of $R$ lie in the interior of the simplex. Furthermore, the entries of the first row of $V$ can be explicitly characterized as follows:
$$
V_{1j} = (\lambda_1+\widetilde{V}^\top_{j*}\mathsf{diag}(\lambda_2,\ldots,\lambda_K)\widetilde{V}_{j*})^{-1/2}.
$$
Observe that the rows of $R$ corresponding to the pure nodes are equal to the rows of $\widetilde V$ and such rows can be identified by running a convex hull vertex finding algorithm on $R$. The rows of $R$ satisfy the linear relation $R=W \widetilde V$, where 
\begin{equation}
\label{eq:omega_deg}
    w_{ij}=\frac{V_{1j} \Pi_{ij}}{\sum_{i=1}^{K}V_{1j} \Pi_{ij}}.
\end{equation} 
If we have prior knowledge of $R$, following the lead of \citet{JIN2023}, we can recover $\Pi$ by solving the linear system $R=W \widetilde V$ to obtain $W$. Then we can apply the reverse transformation on $W$ to recover $\Pi$. In the absence of the prior knowledge of $\Xi$ and hence of $R$, we substitute $\Omega$ by its population counterpart $\pM$ after adjusting for the fact that the diagonals of the matrices are enforced to be zero. Let us consider the leading $K$ eigenvectors of $\pM$ given by $\widehat{\Xi} \in \R^{n \times K}$ and their corresponding eigenvalues $\widehat \lambda_1,\ldots,\widehat \lambda_K$. Now, we define the matrix $\widehat R \in \R^{n \times (K-1)}$ as follows:
\begin{align}
\label{eq:score_norm_eig_vect}
\widehat R_{ij} = \widehat \Xi_{i(j+1)}/\widehat \Xi_{i1}, \qquad \mbox{for $1 \le i \le n$ and $1 \le j \le K-1$.}
\end{align}

\begin{algorithm}[t]
\caption{PriME}\label{alg:priv_mix_det}
\begin{algorithmic}[1]
\Require The privatized adjacency matrix $\mathcal M_\varepsilon(A_{\mathscr G})$; number of communities $K$; the tuning parameters $c$ and $\gamma$.
\State Define: $\pM \leftarrow \frac{1}{1-2p_\varepsilon}\bigg(\M_\varepsilon(A_{\mathscr G})-p_\varepsilon(1_n1^\top_n-I_n)\bigg)$ where $p_\varepsilon=(1+e^\varepsilon)^{-1}$.
\State Set: $d_i \leftarrow \sum_{j \neq i}(\pM)_{ij}$ for $i=1,\ldots,n$ and $\widetilde d^2=\frac{1}{n}\sum_{i=1}^{n}d^2_i$.
\State Compute the $K$ largest eigenvectors of $\pM$. Denote it by $\widehat \Xi$. Let us corresponding eigenvalues be $\widehat \lambda_1,\ldots,\widehat \lambda_K$ in decreasing order.
\State Set $\widehat S = \{i:|\widehat \Xi_{i1}| \ge ~\frac{c\sqrt{\log n}}{(1-2p_\varepsilon)|\widehat{\lambda}_K|}\}$, and $\widehat S_\gamma = \{i \in \widehat S:|\widehat \Xi_{i1}| \ge \gamma\}$.
\State Set $\widehat R_{ij} \leftarrow \widehat \Xi_{i(j+1)}/\widehat \Xi_{i1}$, for $1 \le i \le n$ and $1 \le j \le K-1$.
\State Compute the approximate simplex vertices $\widehat v_2,\ldots, \widehat v_{K}$ by using the Sketched Vertex Search Algorithm on the set of vectors $\{\widehat R_{i*}: i \in\widehat S_\gamma\}$.
\State For $i \notin \widehat S$, set $\widehat \Pi_{i*} \leftarrow \frac{1}{K}1_K$.
\State For $i \in \widehat S$, solve the system of linear equations $\sum_{k=1}^{K}\widehat w_{ik}\widehat v_k=\widehat r_i$  and $\sum_{k=1}^{K}\widehat w_{ik}=1$ to get $(\widehat w_{i1},\ldots,\widehat w_{iK})$.
\State Compute $\widehat{v}_{1j} \leftarrow (\widehat\lambda_1+\widehat{v}^\top_{j}\mathsf{diag}(\widehat\lambda_2,\ldots,\widehat\lambda_K)\widehat{v}_{j})^{-1/2}$ for $j=1,\ldots,K$.
\State Set $\widetilde \Pi_{ik} \leftarrow \max\{\widehat w_{ik}/\widehat{v}_{1k},0\}$, for $1 \le i \le n$ and $1 \le k \le K$.
\State For all $i \in \widehat S$, estimate $\widehat{\Pi}_{i*} \leftarrow \frac{\widetilde \Pi_{i*}}{\|\widetilde \Pi_{i*}\|_1}$.
\Output Estimate of the mixing proportions: $\widehat \Pi$.
\end{algorithmic}
\end{algorithm}

It is worthwhile to mention here that before the pre- and post-PCA normalization, we remove the extremely low-degree nodes from the graph since the information from these nodes is very noisy. In other words, for our analysis, we only consider the nodes in $\widehat S$ where
\begin{equation}
\label{eq:truncation_1}
   {\widehat S = \left\{i:|\widehat \Xi_{i1}| \ge c\frac{(\log n)}{(1-2p_\varepsilon)~|\widehat{\lambda}_K|}\right\}.}
\end{equation}
for a user specified constant $c>0$.
\begin{rem}
    From the construction of $R$, one can observe that the eigenvectors corresponding to the eigenvalues $\lambda_2,\ldots,\lambda_K$ are used for the spectral algorithm whereas the eigenvector corresponding to $\lambda_1$ is mainly used to correct for the effect of $\Theta$. The truncation in \eqref{eq:truncation_1} ensures that the nodes used in the analysis have degrees comparable to the spectral gap of the eigenvalues used in the downstream spectral algorithm. If the degrees are very small compared to the stated spectral gap then their effect cannot be inferred by any spectral algorithm.
\end{rem}

Next, observe that the rows of $\widehat{R}$ corresponding to the pure nodes may not represent the vertices of a simplex because of the noise introduced during the estimation of $\Xi$ by $\widehat \Xi$. Consequently, such vertices cannot be found by running a convex hull vertex search algorithm on $\widehat R$. To address this issue, we employ the \emph{Sketched Vertex Search Algorithm} proposed in Lemma 3.4 of \citet{JIN2023} to approximately identify the rows of $\widehat R$ corresponding to the pure nodes. Further, for the Sketched Vertex Search Algorithm, we only consider the nodes in the set $\widehat S_{\gamma}$ given by
\begin{equation}
    \label{eq:truncation_2}
    \widehat S_\gamma = \{i \in \widehat S:|\widehat{\Xi}_{i1}| \ge \gamma\}
\end{equation}
where $\gamma>0$ is a user-specified constant.

After approximately identifying the rows of $\widehat R$ corresponding to the pure nodes, we can proceed as in the Oracle case to estimate $\Pi$. The detailed procedure for constructing the estimator $\widehat\Pi$ of $\Pi$ is given by Algorithm \ref{alg:priv_mix_det}.

The whole algorithm can be executed in $O(n^{K+1}K^3)$ iterations. This involves the $O(n^{K+1})$ iterations for the \emph{Sketched Vertex Search} and $O(K^3)$ iterations to solve the quadratic optimization which constitutes the two principal steps of the \emph{Sketched Vertex Hunting} procedure.

\subsection{Analysis of Estimation risk of PriME}
In this section, we focus on characterizing the estimation risk of the PriME Algorithm and show that it is rate optimal. In that direction, let us recall the assumptions on the generative model stated in \eqref{eq:singularity_b}, \eqref{eq:spectral_gap}, \eqref{eq:perron_root}, and \eqref{eq:non_negligible_presence}. 

%Observe that the assumption \eqref{eq:non_negligible_presence} ensures that, for each cluster, there is at least one pure node with a significantly large average pseudo-degree. This condition is vital as it guarantees that the use of estimated eigenvectors using $\pM$ efficiently approximates the vertices of the simplex when hunting for such vertices.

 Under the above assumptions, we can establish the following theorem, which characterizes the convergence rate of $\widehat \Pi$ to $\Pi$ under the loss $\mathcal L$ defined in \eqref{eq:loss}. Let us recall $\mathrm{err}_n$ defined in \eqref{eq:err_n_old}.
\begin{thm}
    \label{thm:upper_bound_informal}
Assume that $\log(n)\mathrm{err}_n \rightarrow 0$ as $n \rightarrow \infty$. Furthermore, also assume \eqref{eq:singularity_b}, \eqref{eq:spectral_gap}, \eqref{eq:perron_root}, and \eqref{eq:non_negligible_presence}. Under these assumptions, if $\theta_i\in[0,\,C]$ for all $1\le i\le n$ for some $C>0$ and $\varepsilon \in [0,\varepsilon_0]$ for some $\varepsilon_0>0$, 
then there exists an absolute constant $C_0>0$ (possibly depending on $K$) such that with probability $1-o(n^{-2})$, the estimate $\widehat{\Pi}$ satisfies
\begin{align*}
%\label{eq:rate_loss_error}
\mathcal{L}(\widehat \Pi,\Pi) \le \frac{C_0}{n}\sum_{i=1}^n\min\left\{\frac{\mathrm{err}_n\,(\log n)}{(\theta_i/\widetilde\theta)\wedge  1 },1\right\}=C_0\int_0^\infty\min\left\{\frac{\mathrm{err}_n\,(\log n)}{t\wedge1},\,1\right\}\text{d}F_n(t),
\end{align*}
where $F_n$ is the empirical distribution function of $\{\theta_1/\widetilde\theta,\,\theta_2/\widetilde\theta,\,\ldots,\,\theta_n/\widetilde\theta\}$  with $\widetilde \theta=\|\theta\|_2/\sqrt{n}$.
\end{thm}
%\textcolor{orange}{Comment: Maybe worthwhile to make a comment that our upper bounds holds under conditions weaker than the lower bounds and note that such a phenomena was also observed in the Tracey paper.}
\begin{rem}
    Since the privatized adjacency matrix~$\mathcal{M}_\varepsilon(A_{\mathcal{G}})$ is appropriately de-biased and scaled so that its expectation aligns with the true signal matrix~$\Omega$, the estimator~$\widehat{\Pi}$ can be analyzed under the same set of assumptions as in~\citet{JIN2023}. However, the randomization mechanism inflates the variance of each entry from~$O(\theta_i\theta_j)$ to~$O(1/\varepsilon^2)$, thereby altering the dependence of the rate on the degree distribution relative to the non-private setting. As established in Theorem~\ref{thm:lower_bound_informal}, this slower convergence rate—particularly when~$\widetilde{\theta}$ is close to zero, is unavoidable under local edge differential privacy. Consequently, the proposed privatization mechanism and the resulting estimator are minimax rate-optimal, up to a $\log n$ factor.
\end{rem}
\begin{rem}
   It is also worth noting that while the lower bound requires the degree heterogeneity parameters $\theta_1,\ldots,\theta_n$ to satisfy the conditions of Definition~\ref{def:constr_deg_cor}, the upper bound remains valid even when these conditions are violated, provided that the identifiability assumptions \eqref{eq:singularity_b}, \eqref{eq:spectral_gap}, \eqref{eq:perron_root}, and \eqref{eq:non_negligible_presence} hold. Consequently, the convergence guarantees for the PriME estimates extend to a broader class of networks exhibiting diverse patterns of degree heterogeneity and varying levels of sparsity. However, for networks with highly unbalanced structures—such as star-like graphs with a few influential hub nodes and many low-degree followers (these are precisely the networks where the conditions of Definition~\ref{def:constr_deg_cor} are violated)—the estimation rate of any spectral method may fail to achieve the minimax optimal rate. This phenomenon is also observed in the non-private counterpart.
\end{rem}

%Observing that $\mathcal{L}(\widehat \Pi,\Pi) \le 2$ for all $\Pi$, we immediately get the following corollary to the above theorem.
% \begin{cor}
%     Under the conditions of Theorem \ref{thm:upper_bound_informal}, we have
%     \begin{align}
%        \mathbb E[\mathcal{L}(\widehat \Pi,\Pi)] \le C_0\sqrt{\log n}\int\min\left\{\frac{\mathrm{err}_n}{t \wedge 1},1\right\}\,dF_n(t). 
%     \end{align}
% \end{cor}

\subsubsection{Proof of Theorem \ref{thm:upper_bound_informal}}
To prove Theorem \ref{thm:upper_bound_informal}, we begin by showing that the sample eigenvectors $\widehat{\Xi}$ of $\pM$ concentrate around $\Xi$, the eigenvectors of $\Omega-\mathsf{diag}(\Omega)$ with high probability, where $\Omega=\Theta\Pi B \Pi^\top \Theta$. In that direction, let us consider the following theorem.

\begin{thm}
\label{thm:l_infty_pertb}
Suppose that $\theta_i\in[0,\,C]$ for all $i \in [n]$ for some $C>0$ and $\varepsilon \in [0, \varepsilon_0]$ for some $\varepsilon_0>0.$ Under Assumptions~\eqref{eq:singularity_b}, \eqref{eq:spectral_gap}, \eqref{eq:perron_root}, \eqref{eq:non_negligible_presence}, if
\[
\frac{K^3(\log n)^2}{\beta^2_nn\,\widetilde \theta^4}\left(\frac{e^\varepsilon +1}{e^\varepsilon - 1}\right)^2 \ll 1, \qquad \mbox{for all large values of $n$,}
\]
then there exist $\omega \in \{-1, 1\}$ and an orthogonal matrix $O \in \R^{(K-1) \times (K-1)}$ such that the following holds with probability greater than $1-o(n^{-2})$:
\[
\big|\widehat\Xi_{i1}-\omega\,\Xi_{i1}\big|
\;\lesssim\;
\frac{K\,(\log n)}{(1-2p_\varepsilon)\beta_n\|\theta\|^2_2}\left(1+\frac{\theta_i}{\widetilde \theta}\right),
\]
and
\[
\big\|\widehat\Xi_{i,2:K}-\Xi_{i,2:K}O\big\|_2
\;\lesssim\;
\frac{K^{3/2}\,(\log n)}{(1-2p_\varepsilon)\beta_n\|\theta\|^2_2}\left(1+\frac{\theta_i}{\widetilde \theta}\right),
\]
where $\widetilde{\theta}=\|\theta\|_2/\sqrt{n}$.
\end{thm}
To establish the theorem, we adapt the leave-one-out analysis developed in \citet{ke2022optimal} to our setting. 
The detailed proof is provided in the Supplementary Material. 
We now present a corollary of Theorem~\ref{thm:l_infty_pertb}, which demonstrates that the estimated 
SCORE-normalized eigenvectors $\widehat R$ of the perturbed matrix $\pM$ concentrate around the 
SCORE-normalized eigenvectors $R$ of the population matrix $\Omega - \mathsf{diag}(\Omega)$. 
Specifically, we focus on the subset of nodes whose expected degree parameters $\theta_i$ are sufficiently large. 
For a fixed constant $c_0 > 0$, define
\begin{equation}
    S_n(c_0) = \left\{ i : \theta_i \ge c_0 \, 
    \frac{K\,(\log n)}{(1 - 2p_\varepsilon)\,\beta_n\,\|\theta\|_2} \right\}.
\end{equation}
Note that $S_n(c_0)$ serves as the population analogue of the empirical set $\widehat S$ introduced in 
Algorithm~\ref{alg:priv_mix_det}. 
We are now ready to state the following corollary.

\begin{cor}
\label{cor:concentration_of_R}
    Under the assumptions of Theorem \ref{thm:l_infty_pertb}, for all $i \in S_n(c_0)$ and an orthogonal matrix $O_1 \in \R^{K-1 \times K-1}$, with probability $1-o(n^{-3})$ we have
    $$
    \|O_1\widehat R_{i*}-R_{i*}\|_2 \lesssim \frac{K^{3/2}\,(\log n)}{(1-2p_\varepsilon)\beta_n\sqrt{n}~\widetilde{\theta}(\widetilde{\theta} \wedge \theta_i )}, \quad \mbox{where $\widetilde{\theta}=\|\theta\|_2/\sqrt{n}$.}
    $$
\end{cor}
% \begin{rem}
%     In the implementation of PriME, we estimate $\delta_n$ by $\widehat \delta_n$ defined in \eqref{eq:truncation_1}.
% \end{rem}
Now, let us consider the proof of Theorem \ref{thm:upper_bound_informal}.
\paragraph{Proof of Theorem \ref{thm:upper_bound_informal}}
For $i\notin S_n(c_0)$, we take the trivial estimator $\widehat\pi_i=\frac{1}{K}1_K$ and hence the loss is trivially bounded by a constant. Therefore, we focus on any $i \in S_n(c_0)$. Let us recall that the Sketched Vertex Hunting algorithm operates on vertices of $\widehat S_n(c_0,\gamma)$ which is an approximation to the set 
\[
S_n(c_0,\gamma) = S_n(c_0) \cap \left\{i:\theta_i \ge \gamma \widetilde{\theta}\right\} \quad \mbox{where $\widetilde{\theta}=\|\theta\|_2/\sqrt{n}$.}
\]
Therefore, using the property of Sketched Vertex Hunting algorithm (Lemma E.3 of \cite{JIN2023}) and Corollary \ref{cor:concentration_of_R}, we get a permutation matrix $\mathsf P \in \R^{K \times K}$ such that for an orthogonal matrix $O_1$
\begin{align*}
\|\mathsf P\widehat VO-V\|_{2 \rightarrow \infty} &\lesssim \max_{i \in S_n(c_0,\gamma)} \|O_1\hat{R}_{i*}-R_{i*}\|_2\\
& \lesssim_\gamma \frac{K^{3/2}\sqrt{n\,(\log n)^2}}{(1-2p_\varepsilon)\beta_n\|\theta\|^2_2} \ll 1,
\end{align*}
where $V \in \R^{K \times K}$ are the vertices of simplex enclosing the true $R$ as defined in Section \ref{def_alg_terms} and $\widehat{V}$ are their empirical analogues constructed using vertex hunting. 
Proceeding as in the proof of Theorem 4.2 of \cite{ke2022optimal}, we can show using the above bound that with probability greater than $1-o(n^{-2})$, we have a permutation matrix $T \in \R^{K \times K}$ such that
$$
\|T\widehat\pi_i -\pi_i\|_1 \lesssim \min\left\{\frac{K^{3/2}\,(\log n)}{(1-2p_\varepsilon)\beta_n\sqrt{n}~\widetilde{\theta}(\widetilde{\theta} \wedge \theta_i )},1\right\},
$$
simultaneously for all $i \in [n]$. Finally, using the definition of $\mathcal L(\cdot,\cdot)$ from \eqref{eq:loss} and $\mathrm{err}_n$ from \eqref{eq:err_n_old}, the stated result follows by recognizing that $F_n$ is the empirical distribution of $\{\theta_1/\widetilde\theta,\ldots,\theta_n/\widetilde\theta\}$.

\section{Experiments}
\label{numerical}
\begin{figure}[t]
    \centering
    \begin{subfigure}[b]{0.49\textwidth}
        \centering
        \includegraphics[width=\columnwidth]{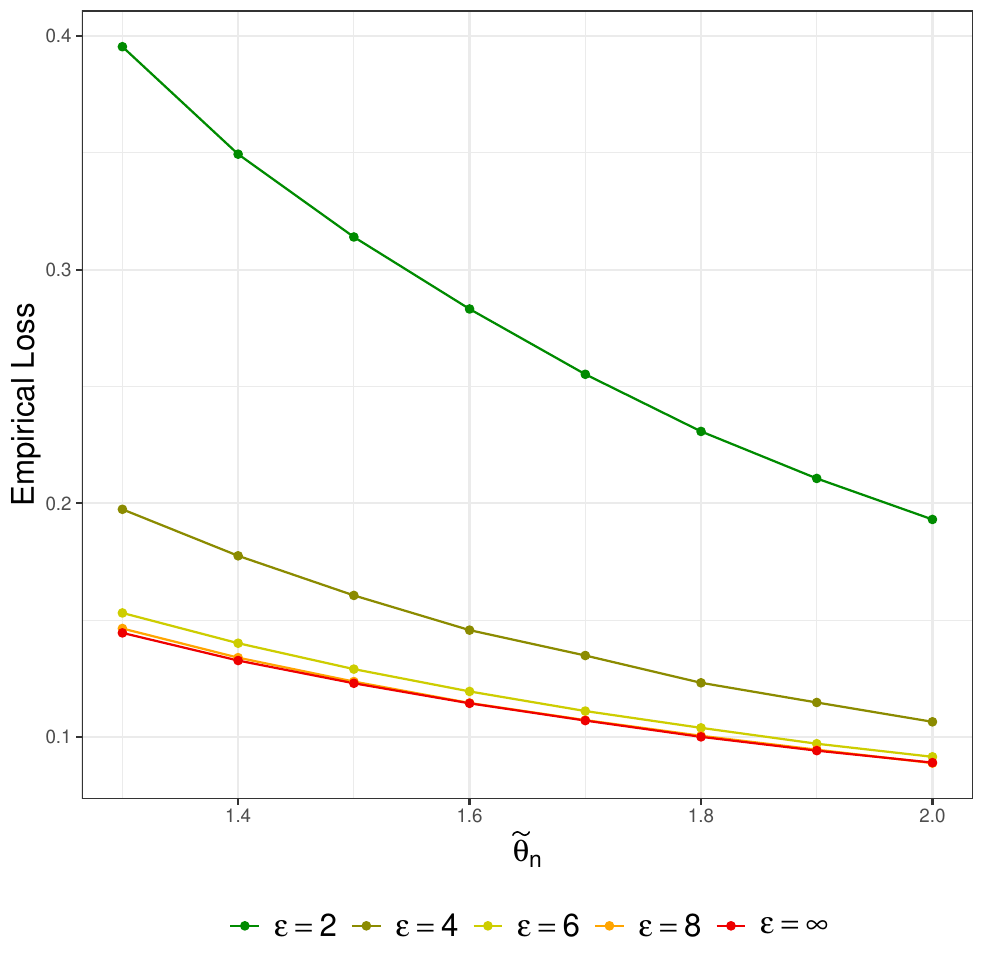}
        \caption{$\mathbb E[\mathcal L(\widehat \Pi,\Pi)]$ as a function of $\widetilde \theta$. }
        \label{fig:snr}
    \end{subfigure}%
    ~
    \begin{subfigure}[b]{0.49\textwidth}
        \centering
        \includegraphics[width=\columnwidth]{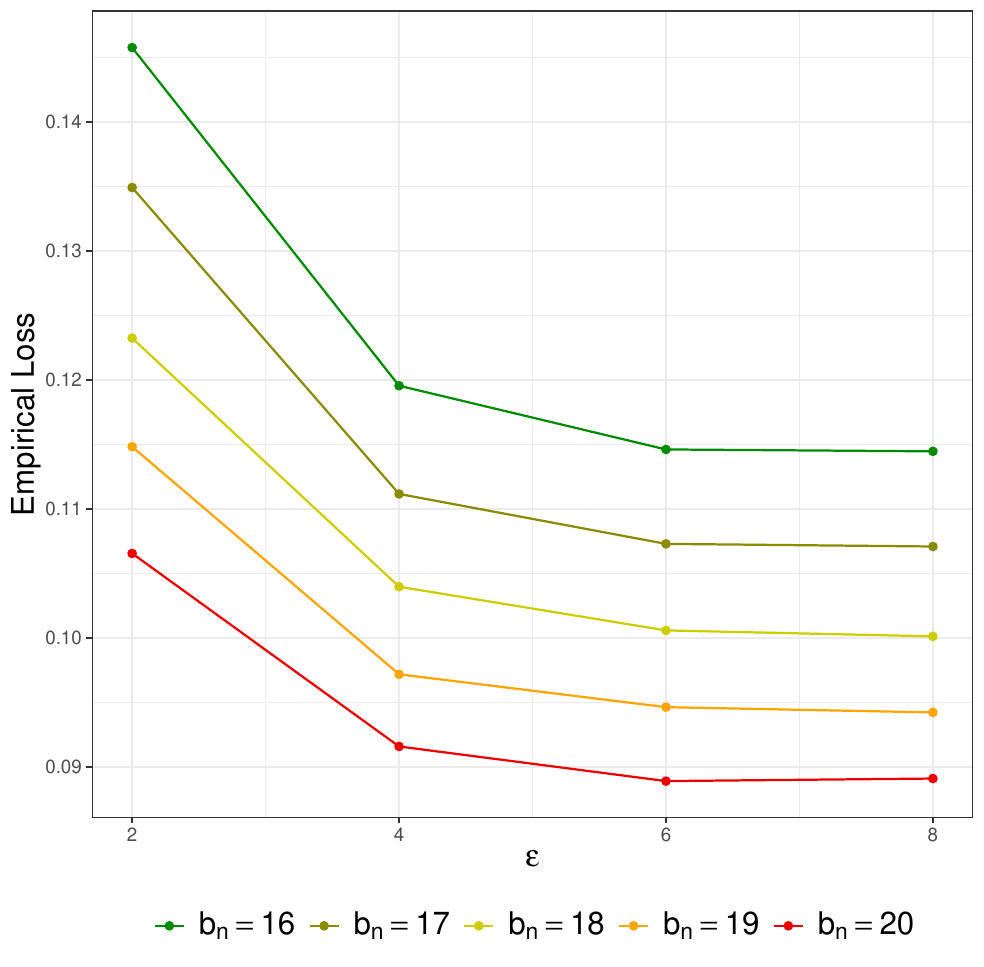}
        \caption{$\mathbb E[\mathcal L(\widehat \Pi,\Pi)]$ as a function of $\varepsilon$. }
        \label{fig:eps}
    \end{subfigure}%
    \caption{Plot of average loss $\mathcal{L}(\widehat \Pi,\Pi)$ over 100 independent replications as functions of $\widetilde \theta$ and $\varepsilon$.}
\end{figure}

To investigate how interplay between the privacy requirements specified by $\varepsilon$ and the degree heterogeneity in the network characterized by $\theta_1,\ldots,\theta_n$ impacts the accuracy of PriME, we conducted a series of numerical simulations. Our study revolved around a random graph with 2000 vertices and two communities, generated as per the model defined in Equation~\eqref{eq:DCMM}. To set the degree correction parameters $\theta_1, \ldots, \theta_n$, we adopted a specific procedure: we initially drew values $\theta^\circ_1, \ldots, \theta^\circ_n$ from a uniform distribution on the interval $[0.3,\,5]$. Subsequently, we set $\theta_i=b_n\theta^\circ_i/(\sum_{i=1}^{n}(\theta^\circ_i)^2)^{1/2}$ to ensure that $b^2_n=n\widetilde\theta^2=\|\theta\|_2^2$. By generating the degree heterogeneity parameters from a distribution with high entropy we ensure that the generated network has high degree heterogeneity.
Furthermore, the connection probability matrix is set to be $B=\xi_nI_n+(1-\xi_n)1_n1^\top_n$, with $\xi_n=0.9$. The tuning parameters of the algorithm were set to be $c=0.005$, and $\gamma=0.02$ throughout all numerical experiments.

The following experiments were designed to investigate how the sparsity level of the network—captured by the aggregated degree-heterogeneity parameter $\widetilde\theta$ and the privacy requirement $\varepsilon$—affects performance.

\begin{itemize}
  \item \textbf{Effect of sparsity on accuracy:} In our first experiment, we explored the impact of varying $\widetilde \theta=\|\theta\|_2^2/\sqrt{n}$ on the expected estimation risk $\mathbb E[\mathcal{L}(\widehat \Pi,\Pi)]$. Observe that, smaller values of $\widetilde\theta$ correspond to sparser graphs. To achieve this, we performed simulations on random networks based on the model specified above, with $\|\theta\|_2:=b_n \in \{13,\ldots,20\}$. For each value of $b_n$, we executed PriME in 100 replications. We estimated the average $\mathcal{L}(\widehat \Pi,\Pi)$ over these 100 replications to gauge the estimation risk. We performed this experiment for $\varepsilon\in\{2,\,4,\,6,\,8\}.$ We also executed the MIXED-SCORE algorithm from \cite{JIN2023} on the same simulation setting to compare with the non-private ($\varepsilon=\infty$) case. The resulting behavior of $\mathbb E[\mathcal L(\widehat \Pi,\Pi)]$ as a function of $\widetilde \theta=b_n/\sqrt{n}$ is visually presented in Figure \ref{fig:snr}. We see the accuracy is increasing with an increase in $\widetilde \theta$. This implies that when the network is less sparse, indicated by relatively higher values of $\widetilde \theta$, the estimation accuracy of the membership profiles for the nodes increases. This conforms to our findings in Theorem \ref{thm:upper_bound_informal}. Observe that it is well understood that spectral clustering methods perform well when the underlying network is dense. When the communities are pure, there are alternative techniques that can be used to boost the accuracy of the recovery of community labels. However, as shown in Theorem \ref{thm:lower_bound_informal}, under the mixed membership set-up, the membership profile recovery in sparse networks is fundamentally hard.

  \item \textbf{Effect of Privacy on Accuracy:} For the second experiment, we varied the privacy parameter $\varepsilon\in\{2,\,4,\,6,\,8\}$. For each fixed value of $\varepsilon,$ we executed PriME in 100 replications. Then we estimated the average $\mathcal{L}(\widehat\Pi,\Pi)$ over these 100 replications. We performed this experiment for $b_n\in\{16,\,17,\,\ldots,\,20\}.$ In Figure \ref{fig:eps}, we illustrate the behavior of $\mathbb E[\mathcal{L}(\widehat \Pi,\Pi)]$ in response to different privacy levels, providing insights into the trade-off between privacy requirements and estimation accuracy. We observe that accuracy is increasing as privacy becomes weaker ($\varepsilon$ increases), which aligns with our findings in Theorem \ref{thm:upper_bound_informal}.
\end{itemize}

% \begin{figure}[!ht] \centering 
% \centering \includegraphics[width=0.5\columnwidth]{Figures/loss_vs_bn.pdf} \caption{Plot of $\mathbb E[\mathcal L(\widehat \Pi,\Pi)]$ as a function of $b_n$ for fixed values of $\varepsilon$. }
% \label{snr}
% \end{figure}

% \begin{figure}[!ht] \centering 
% \centering \includegraphics[width=0.5\columnwidth]{Figures/loss_vs_eps.pdf} \caption{Plot of $\mathbb E[\mathcal L(\widehat \Pi,\Pi)]$ as a function of $\varepsilon$ for fixed values of $b_n$.}
% \label{eps}
% \end{figure}

\section{Real Data Examples}
\label{real_data}

In this section, we illustrate the performance of our method in identifying communities within real-world networks.
%\textcolor{orange}{Comment: Best to switch the order of these data sets. Starting with the Facebook Ego data seems better because there are actual privacy concerns there compared to the political blogs data.}

\subsection*{Facebook Ego Network Dataset}

%\textcolor{orange}{Highlight the privacy risk in this section. Also, it's not clear what is meant by ego vertex or ego circle. }
Next, we apply our algorithm to recover the membership profiles of individuals in different social circles by studying their friendship networks on Facebook. In that direction, we consider the Facebook Ego Network dataset from \citet{NIPS2012_7a614fd0}. The data is openly available at \url{http://snap.stanford.edu/data/ego-Facebook.html}.

The dataset consists of multiple ego networks, each defined by a central \emph{ego} node, its direct connections (\emph{alters}), and the edges among these alters. For our analysis, we selected four ego networks and extracted the induced subgraphs, ensuring in each case that the resulting graph is connected.

Maintaining privacy is a fundamental concern when analyzing connectivity data of this nature. Indeed, Facebook restricted public release to ego networks precisely because only a small number of users consented to sharing information about their friendship ties. If an adversarial user were able to query the recovered membership profiles—obtained by running a mixed-membership estimation algorithm on such networks—to infer sensitive attributes of individuals, it would further discourage users from disclosing their connections. This erosion of trust ultimately hinders the ability of researchers to conduct reliable, population-level analysis using social network data.

\begin{figure*}[t]
    \centering
    \includegraphics[width = 0.8\textwidth]{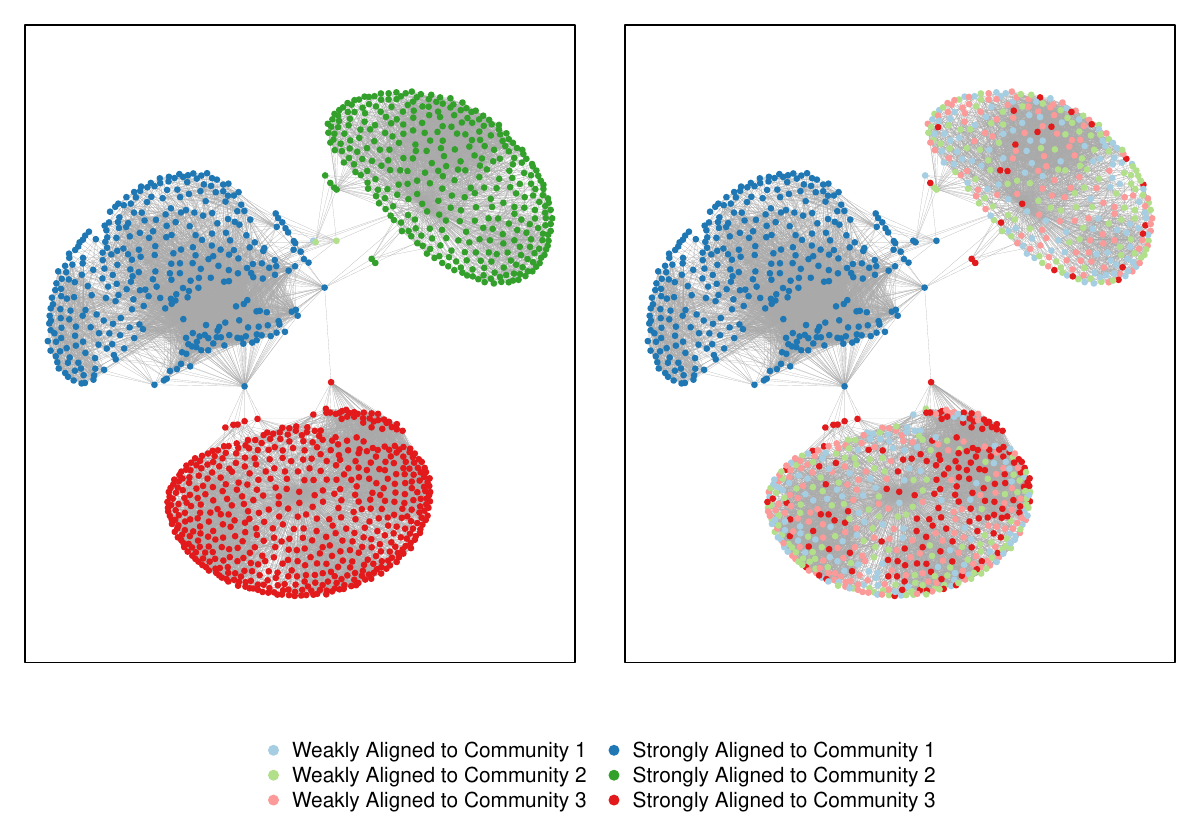}
    \caption{\small Facebook ego network with different communities.
    Left panel: Communities estimated by MIXED-SCORE from \cite{ke2022optimal}.
    Right panel: Communities estimated by {PriME} with $\varepsilon = 8$.
    }
    \label{fig:fb}
\end{figure*}

PriME provides an alternative framework in which the data provider applies a randomization mechanism to privatize edge information before releasing the network. This allows analysts to conduct meaningful inference while ensuring that individual-level connections remain protected. Inevitably, however, enforcing privacy affects the downstream accuracy of membership-profile estimation. In particular, the quality of the estimates depends on both the node’s degree and the underlying signal strength in the network. Through our analysis of this dataset, we investigate this interplay between estimation accuracy and the imposed privacy constraints.

%Whereas the other nodes are weakly or strongly associated with one or more of such ego circles. We chose a set of 4 ego vertices and considered the subgraph formed by them and their neighbors for our analysis. 
In the final data set that we analyzed, there were $n=1234$ vertices and 12913 edges. We executed the MIXED-SCORE from \cite{JIN2023} algorithm and PriME on this dataset with the number of clusters equal to $3$ and $\varepsilon=8$. Each vertex $i$ was assigned to the cluster $j$ where $j=\argmax_{1\le j\le K}\widehat{\Pi}_{ij}$. The ties were broken at random. To investigate the alignment strength of the vertices to different communities, we consider the metric $V(i)=\max_{1\le j\le K}\widehat\Pi_{ij}$ for $1\le i\le n.$ This metric captures how ``pure" a vertex is, with the pure nodes having $V(i)$ close to $1$ and the impure nodes having $V(i)$ close to $1/3$. We categorize a vertex to be weakly aligned to its assigned cluster if $V(i)\in[1/3,\,1/2]$ and strongly aligned if $V(i)\in(1/2,\,1].$ Figure \ref{fig:fb} demonstrates the community structure of the network where the vertices in different communities are colored according to the strength of their alignment to their assigned communities. 

From our analysis, we find that enforcing privacy does not noticeably affect the quality of the estimated membership profiles for vertices in Community 1. In contrast, for the other two communities, the privatization step leads to substantial degradation: the inferred community labels become less confident, and misclassification becomes more frequent, with Community 2 being particularly affected. This behavior is expected in networks with pronounced degree heterogeneity and weak signal strength—settings in which privatizing edge information significantly distorts the underlying graph structure. As a result, the membership profiles of a large fraction of nodes become inherently difficult to estimate accurately.

%The distance between $\widehat \Pi_\varepsilon$ and $\widehat{\Pi}_\infty$ in terms of $\mathcal L(\cdot,\cdot)$ defined in \eqref{eq:loss} is found to be $0.1168$. 

\begin{figure*}[t]
    \centering
    \includegraphics[width = 0.8 \textwidth]{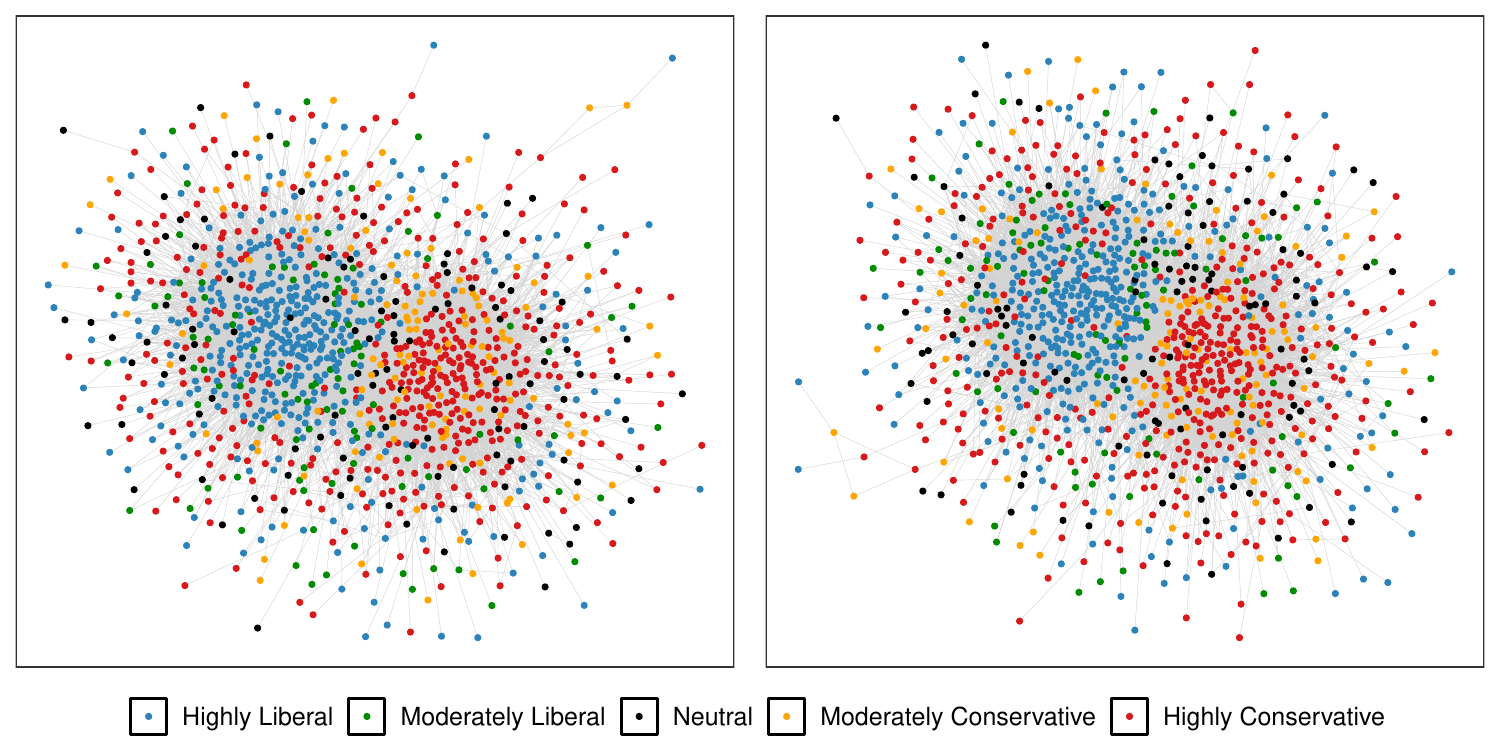}
    \caption{\small The political blog network and political affiliations of the blogs.
    Left panel: Communities estimated by MIXED-SCORE from \cite{ke2022optimal}.
    Right panel: Communities estimated by PriME with $\varepsilon = 1.5$. Out of $1222$ political blogs, 33 blogs had different group memberships in private and non-private cases.
    }
    \label{fig:blogs}
\end{figure*}

\begin{figure}[t]
    \centering
    \includegraphics[scale=0.4]{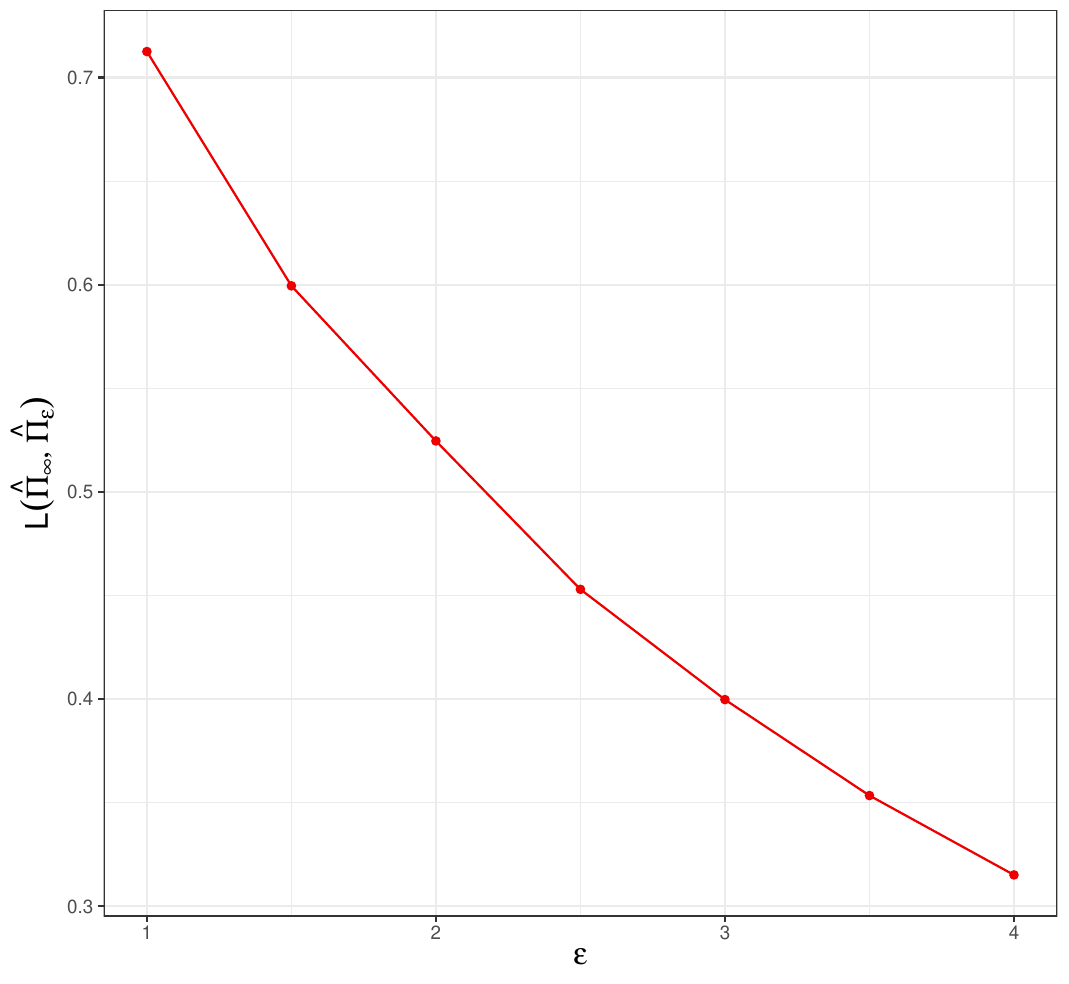}
    \caption{Plot of $\mathcal{L}(\widehat\Pi_\infty,\,\widehat\Pi_\varepsilon)$ as a function of $\varepsilon$ in political blog data.}
    \label{fig:loss against epsilon}
\end{figure}

\subsection*{Political Blog Dataset}
Next, we applied our algorithm to detect communities within the Political Blogs data from \citet{Adamic2005ThePB}. This dataset records linkages between political blogs in a single-day snapshot before the 2004 US Presidential Elections. The nodes of the network are $1222$ political blogs and an edge exists between two nodes if one of the nodes cites another. While privacy is not a primary concern for this particular dataset, it serves as an illustrative example of our method. It is important to note that in several other scenarios, when collecting data on political linkages, privacy concerns may be significant and individuals from whom the data was collected may want assurance of privacy protection. In such situations, our method will be useful for analyzing the data while maintaining privacy constraints. In this data, the nodes are already classified into \emph{conservative} or \emph{liberal}. We estimate to estimate the affiliation of each blog towards the liberal spectrum using PriME with $\varepsilon=1.5$ and the non-private MIXED-SCORE algorithm from \cite{JIN2023}. Then the blogs are classified into `highly conservative', `moderately conservative', `neutral', `moderately liberal', and `highly liberal' based on whether their membership likelihood for the liberal community is between $[0,0.2], (0.2,0.4], (0.4,0.6], (0.6,0.8]$ and $(0.8,1]$ respectively. In Figure \ref{fig:blogs}, we provide a visual representation of how the nodes in this network are affiliated with these five political categories.

Finally, in Figure \ref{fig:loss against epsilon}, we show the distance between $\widehat \Pi_\varepsilon$ and $\widehat{\Pi}_\infty$ in terms of $\mathcal L(\cdot,\cdot)$ defined in \eqref{eq:loss} as a function of $\varepsilon$, where $\widehat \Pi_\varepsilon$ and $\widehat{\Pi}_\infty$ are the estimators of $\Pi$ in the private and the non-private settings, respectively. Through this plot, we demonstrate the effect of enforcing privacy on the estimation procedure for $\Pi$. As expected, reducing the privacy requirements (increasing $\varepsilon$) reduces the difference between the non-private estimation procedure laid out in \cite{JIN2023} and PriME.

However, because this dataset exhibits relatively high signal strength and limited degree heterogeneity, the performance gap between PriME and MIXED-SCORE remains minimal. This experiment highlights that the effect of privacy on the accuracy of membership-profile estimation is highly sensitive to the structural properties of the underlying network—particularly its signal strength, sparsity, and degree heterogeneity.

\section{Discussion}

This paper introduces PriME, a novel approach that addresses the fundamental challenge of estimating community membership probabilities within the Degree Corrected Mixed Membership Stochastic Block Model, while ensuring individual edge privacy. Operating under the  $\varepsilon$-edge Local Differential Privacy framework, our algorithm leverages a symmetric edge flip strategy to release a privatized version of the network for downstream analysis. The development of PriME—a spectral algorithm meticulously designed to estimate the membership matrix $\Pi$ while adhering to $\varepsilon$-local differential privacy constraints—highlights our method's practical and theoretical robustness. This study pioneers the application of privacy constraints within mixed membership models, marking it as the first of its kind and positioning our approach at the forefront of privacy-aware community detection.

Our findings reveal that PriME consistently produces estimators that not only have robust practical performance but achieve minimax optimal rates in estimating the membership matrix. This performance is corroborated through extensive numerical simulations and real-world data applications, highlighting the algorithm's utility and reliability in actual network environments.

Despite these advancements, our work also opens several avenues for further research. Currently, our focus has been predominantly on edge differential privacy. However, the literature, including works like those by \cite{borgs2018revealing}, suggests the potential benefits of node differential privacy, which offers more robust protections. Exploring rate-optimal algorithms for community recovery that incorporate such stronger privacy measures remains an unexplored area and presents a fertile ground for future research.

Additionally, considering other forms of privacy, such as central differential privacy and investigating the information-theoretic lower bounds required to consistently detect underlying community structures under various privacy constraints are compelling directions for research. Moreover, establishing lower bounds for the risk associated with estimating community labels, particularly in scenarios where labels are not mixed, would significantly enhance our understanding of privacy-preserving mechanisms in network analysis.

In conclusion, the introduction of PriME represents a significant contribution to the ongoing evolution of privacy-aware community detection. By addressing previously unexplored challenges in mixed membership models under strict local differential privacy constraints, our work lays a solid foundation for future research in understanding the cost of enforcing privacy constraints under a mathematically rigorous framework for complex data analysis problems involving relational data represented through random networks.

\appendix
\section{Proof of Theorem \ref{thm:lower_bound_informal}}
Let us recall that we assume $\theta_i=O(1)$ for all $i \in [n]$ and the privacy parameter $\varepsilon \in [0,\varepsilon_0]$. Let us define $\nu_i:=\theta_i/\widetilde\theta$ where $\widetilde \theta=\|\theta\|_2/\sqrt{n}$.

\subsection{Proof of Theorem \ref{thm:lower_bound_informal} under (\ref{eq:F_n-condn})}
% Let us consider the following property of $F_n$, the empirical distribution of $\theta_i/\widetilde\theta$ for $i \in [n]$.
% \begin{lem}[Lemma 5.2, \cite{ke2022optimal}]
% \label{lem:imp_lem_lower_bound}
%     For fixed $\gamma,\,a_0\in(0,\,1)$ and $\theta\in\mathcal{G}_n(\gamma,\,a_0),$ we denote $\nu_i=\theta_i/\widetilde\theta,\,1\le i\le n.$ Let $\widetilde F_n$ be the empirical distribution of $\widetilde\nu_i=\nu_i\wedge 1,\,1\le i\le n.$ For any $c>0$ and $\epsilon\in(0,\,1),$ define
%     $$\tau_n(c,\,\epsilon)=\inf\left\{x>0: \left|\left\{i:x <\widetilde \nu_i \le c\right\}\right| \le \varepsilon\left|\left\{i:\mathrm{err}_n <\widetilde \nu_i \le c\right\}\right|\right\}.$$
%     If $F_n$ satisfies \eqref{eq:F_n-condn}, then there exists a number $c_n>\mathrm{err}_n$ and a constant $\widetilde a_0\in(0,\,1)$ such that $F_n(c_n)\le1-\widetilde a_0$ and
%     \begin{equation}
%     \sum_{i=1}^n\frac{1}{\widetilde\nu_i }\mathbbm 1\left\{\tau_n(c_n,\,1/8) <\widetilde\nu_i\le c_n\right\}
%         +\frac{\lceil n\omega_n\rceil}{n(\tau_n(c_n,\,1/8)\wedge1)}\ge\widetilde a_0\sum_{i=1}^n\frac{1}{\widetilde\nu_i}\mathbbm 1\left\{\widetilde \nu_i>\mathrm{err}_n\right\},
%         \label{eq:F_n}
%     \end{equation}
%     where 
%     \[
%     \omega_n=\frac{\left|\left\{i:\mathrm{err}_n<\widetilde\nu_i\le c_n\right\}\right|}{8n}-\frac{\left|\left\{i:\tau_n(c_n,\,1/8)<\widetilde\nu_i\le c_n\right\}\right|}{n}.
%     \]
%     \label{lem:F_n}
% \end{lem}
% \nb{Sayak to write a proof}\newline
\paragraph{Construction of a least favorable configuration}
%The proof of this lemma can be found in Appendix E.2 of \cite{ke2022optimal}.
We begin by constructing a least favourable configuration for the pair of parameters $(\Pi,B)$ for a fixed configuration of average degree parameters $\theta \in \mathcal G_n(\rho,a_0)$. Let us first order $\theta_i$'s such that we have $\theta_{(1)}\le\theta_{(2)}\le\ldots\le\theta_{(n)}.$ We also consider the corresponding ranked $\nu_i$ parameters  $\nu_{(1)}\le\nu_{(2)}\le\ldots\le\nu_{(n)}$. Consider $c_n$ as in Definition \ref{def:constr_deg_cor} and define
$$s_n=\max\{1\le i\le n:\nu_{(i)}\le\mathrm{err}_n\}\quad\text{and}\quad n_0=\max\{1\le i\le n:\nu_{(i)}\le c_n\}-s_n.$$
This means $\mathrm{err}_n\le\rho c_n\le c_n$, where $\mathrm{err}_n$ is defined in \eqref{eq:err_n_old} and
\begin{equation}
    \int_{\rho c_n}^{c_n}\frac{\text{d}F_n(t)}{t\wedge 1}\ge a_0\int_{\mathrm{err}_n}^{\infty}\frac{\text{d}F_n(t)}{t\wedge 1},
    \label{eq:defn2.2}
\end{equation}
where $F_n$ is the empirical distribution of $\nu_1,\ldots,\nu_n$.
Moreover $n_0$ is the number of $\nu_{(i)}$'s such that $\mathrm{err}_n<\nu_{(i)}\le c_n.$ Here $c_n$ has to be bounded by some constant $C>0.$ This follows from the fact that $\int_0^1t^2\text{d}F_n(t)=1$ and \eqref{eq:defn2.2}.
Define
\begin{equation}
    \gamma_n:=\frac{c_0\sqrt{K}}{\beta_n\sqrt{n}~\widetilde \theta}\left(\frac{e^\varepsilon +1}{e^\varepsilon - 1}\right)=\frac{c_0\,\widetilde\theta}{K}\mathrm{err}_n.
    \label{eq:gamma_n}
\end{equation}
for some constant $c_0$ to be chosen later. Let $\mathcal{N}_0$ be the set of indices such that after re-ordering the $\theta_i$'s, the indices are between $s_n$ and $s_n+n_0.$ Now re-arrange the nodes by putting the nodes in $\mathcal{N}_0$ as the first $n_0$ nodes. Then, for $1\le i\le n_0,$ we have $\mathrm{err}_n<\theta_i/\widetilde\theta\le c_n\le C.$ The last $n-n_0$ nodes will be pure nodes distributed among the $K$ communities such that the sum of squares of degrees of the pure nodes in different communities are of the same order. In particular, if $\mathcal{C}_k$ denotes the pure nodes in community $k,$ we want to make sure that
\begin{equation}
    \min_{1\le k\le K}\sum_{i\in\mathcal{C}_k}\theta_i^2\ge \widetilde{c}\max_{1\le \ell\le K}\sum_{i\in\mathcal{C}_\ell}\theta_i^2
    \label{eq:deg_condn}
\end{equation}
for some constant $\widetilde c.$ Let $\{e_1,\,e_2,\,\ldots,\,e_K\}$ be the Euclidean basis vectors of $\R^K.$ We set
$$\Pi^{(0)}=\left(\underbracket{\frac{1_K}{K},\,\frac{1_K}{K},\,\ldots,\,\frac{1_K}{K}}_{n_0},\,\underbracket{e_1,\,e_1,\,\ldots,\,e_1,\,e_2,\,e_2,\,\ldots,\,e_2,\,\ldots,\,e_K,\,e_K,\,\ldots,\,e_K}_{n-n_0}\right)^\top;$$
where $e_k$ is the $k$-th canonical basis vector for $\R^K$.
Observe that the first $n_0$ of the $\theta_i$'s are assigned to the mixed nodes and for $1\le i\le n_0,$ using $\mathrm{err}_n \ll 1$, we have $\theta_i/\widetilde\theta\ge\mathrm{err}_n.$ Now let $m=\lfloor n_0/2\rfloor,\,r=\lfloor K/2\rfloor$ and $s=mr.$ Then by Varshamov-Gilbert Lemma (See Lemma 2.9 of \cite{MR2724359}), we can find binary vectors $w^{(0)},\,w^{(1)},\,\ldots,\,w^{(J)}\in\{0,\,1\}^s$ such that $J\ge 2^{s/8},\,w^{(0)}=0$ and $\|w^{(j)}-w^{(k)}\|_1\ge s/8$ for all $0\le j<k\le J$, and rearrange each $w^{(j)}$ to form the matrices $\widetilde{\Gamma}^{(j)} \in \R^{m\times r}$ for $0\le j\le J.$ Finally, define $\Gamma^{(j)} \in \R^{n\times K}$ as follows:
\begin{align}
\label{eq:construct_gamma}
  \Gamma^{(j)}=\left[
    \begin{array}{@{} c c c@{}}
     \widetilde{\Gamma}^{(j)}&-\widetilde{\Gamma}^{(j)}&0_{m\times 1}\\
     -\widetilde{\Gamma}^{(j)}&\widetilde{\Gamma}^{(j)}&0_{m\times 1}\\
    0_{1\times r}&0_{1\times r}&0\\
      \cdashline{1-3}
      0_{(n-n_0)\times r}&0_{(n-n_0)\times r}&0_{(n-n_0)\times 1}
    \end{array}
  \right],\quad0\le j\le J.
\end{align}
If $K$ is even, the last column of $\Gamma^{(j)}$ disappears and if $n_0$ is even, the last row above the dotted line of $\Gamma^{(j)}$ disappears. Denote $\Theta_{\mathrm{trunc}}=\mathsf{diag}(\theta_1\wedge\widetilde\theta,\,\theta_2\wedge\widetilde\theta,\,\ldots,\,\theta_n\wedge\widetilde\theta) \in \R^{n \times n}$. Now, construct
\begin{equation}
\Pi^{(j)}=\Pi^{(0)}+\gamma_n\Theta_{\mathrm{trunc}}^{-1}~\Gamma^{(j)},\quad\text{for }0\le j\le J,
    \label{eq:lfc-mem}
\end{equation}
and consider the connectivity matrix given by $B^*=c'\beta_nI_K+(1-c'\beta_n)1_K1_K^\top$ for some constant $c'\in(0,\,1).$ First, we show that for all $0\le j\le J,$ the matrices $(\Pi^{(j)},\,B^*)\in\mathcal{Q}_n(\theta)$, where $\mathcal{Q}_n(\theta)$ is the set of matrices $(\Pi,B)$ satisfying \eqref{eq:singularity_b}, \eqref{eq:spectral_gap}, \eqref{eq:perron_root},  \eqref{eq:non_negligible_presence}.

Firstly, $\Pi^{(0)}$ is a valid membership matrix with all entries in $[0,\,1]$ and rows summing up to 1. By the definition of the perturbation matrices, in $\Pi^{(j)},$ only the rows $1\le i\le n_0$ are perturbed. Before perturbation, each of the entries in the first $n_0$ rows were $\frac{1}{K}.$ Now recalling the definition of $\gamma_n$ in \eqref{eq:gamma_n}, and the fact that $\theta_i/\widetilde\theta\ge\mathrm{err}_n$ for $1\le i\le n_0,$ we can bound the magnitude of the perturbations in \eqref{eq:lfc-mem} by
$$\frac{\gamma_n}{\theta_i\wedge\widetilde\theta}=\frac{c_0\widetilde\theta\mathrm{err}_n}{K(\theta_i\wedge\widetilde\theta)}=\frac{c_0\mathrm{err}_n}{K((\theta_i/\widetilde\theta)\wedge1)}\le\frac{c_0\mathrm{err}_n}{K(\mathrm{err}_n\wedge1)}\le\frac{c_0}{K},$$
since $\mathrm{err}_n\ll1.$ Now $c_0$ can be chosen small enough that even after the perturbations, $\Pi^{(j)}$ are valid membership matrices with rows adding up to 1. Note that the first condition in \eqref{eq:singularity_b} is trivially true since $\min_{i,j}B^*_{ij}=1-c'\beta_n\ge c$ when $c'$ is chosen appropriately. Denote
$$G^{(j)}=K\|\theta\|_2^{-2}({\Pi^{(j)}}^\top \Theta^2\Pi^{(j)})\quad\text{for }0\le j\le J.$$
Then,
\begin{align*}
    G^{(0)}&=\frac{K}{\|\theta\|_2^2}\left[\frac{1}{K^2}\left(\sum_{i=1}^{n_0}\theta_i^2\right)1_K1_K^\top+\sum_{k=1}^K\left(\sum_{i\in\mathcal{C}_k}\theta_i^2\right)e_ke_k^\top\right]\\
    &=\frac{1}{K\|\theta\|_2^2}S_01_K1_K^\top+\frac{K}{\|\theta\|_2^2}\mathrm{diag}(S_1,\,S_2,\,\ldots,\,S_K)
\end{align*}
where $S_0=\sum_{i=1}^{n_0}\theta_i^2$ and $S_k=\sum_{i\in\mathcal{C}_k}\theta_i^2.$ Thus,
$$\|G^{(0)}\|\le\frac{S_0}{\|\theta\|_2^2}+\frac{K}{\|\theta\|_2^2}\max_{1\le k\le K}S_k\le\frac{S_0}{\|\theta\|_2^2}+\frac{\widetilde cK}{\|\theta\|_2^2}\min_{1\le k\le K}S_k\lesssim 1,$$
due to \eqref{eq:deg_condn} and the fact that $S_0+\sum_{k=1}^KS_k=\|\theta\|_2^2.$ Similarly, the minimum eigenvalue of $G^{(0)}$ is bounded below by
\begin{equation}
    \lambda_{\min}(G^{(0)})\ge\frac{K}{\|\theta\|_2^2}\min_{1\le k\le K}S_k.
    \label{eq:G_min}
\end{equation}
Recall that for $1\le i\le n_0,$ we have $\theta_i/\widetilde\theta\le c_n\le C$ for some $C>0.$ Then
$$n\widetilde\theta^2=\|\theta\|_2^2=S_0+\sum_{k=1}^KS_k\le Cn_0\widetilde\theta^2+K\widetilde c\min_{1\le k\le K}S_k.$$
Thus $K\min_{1\le k\le K}S_k\asymp\|\theta\|_2^2.$ Plugging this into \eqref{eq:G_min}, we obtain $\lambda_{\min}(G^{(0)})\gtrsim 1.$ This ensures the last condition in \eqref{eq:singularity_b} for $G^{(0)}.$ Now
\begin{align}
    \|G^{(j)}-G^{(0)}\|&=\frac{K}{\|\theta\|_2^2}\|(\Pi^{(0)}+\gamma_n\Theta_{\mathrm{trunc}}^{-1}\Gamma^{(j)})^\top\Theta^2(\Pi^{(0)}+\gamma_n\Theta_{\mathrm{trunc}}^{-1}\Gamma^{(j)})-{\Pi^{(0)}}^\top\Theta^2\Pi^{(0)}\|\nonumber\\
    &\le\frac{2K}{\|\theta\|_2^2}\|G^{(0)}\|^{1/2}\|\gamma_n^2{\Gamma^{(j)}}^\top\Theta_{\mathrm{trunc}}^{-1}\Theta^2\Theta_{\mathrm{trunc}}^{-1}\Gamma^{(j)}\|^{1/2}\nonumber\\
    &\quad\quad\quad\quad\quad\quad\quad\quad\quad\quad+\frac{K}{\|\theta\|_2^2}\|\gamma_n^2{\Gamma^{(j)}}^\top\Theta_{\mathrm{trunc}}^{-1}\Theta^2\Theta_{\mathrm{trunc}}^{-1}\Gamma^{(j)}\|,
    \label{eq:G_gap}
\end{align}
where we use the fact that $\|A^\top MB+B^\top MA\|\le2||A^\top MA\|^{1/2}||B^\top MB\|^{1/2}$ for any positive semidefinite matrix $M.$ Moreover,
\begin{equation}
    \frac{K}{\|\theta\|_2^2}\|\gamma_n^2{\Gamma^{(j)}}^\top\Theta_{\mathrm{trunc}}^{-1}\Theta^2\Theta_{\mathrm{trunc}}^{-1}\Gamma^{(j)}\|=\frac{c_0^2\mathrm{err}_n^2}{nK}\|{\Gamma^{(j)}}^\top\Theta_{\mathrm{trunc}}^{-1}\Theta^2\Theta_{\mathrm{trunc}}^{-1}\Gamma^{(j)}\|=o(1)
    \label{eq:perturb_bound}
\end{equation}
recalling the definition of $\gamma_n$ in \eqref{eq:gamma_n}. Now, by Weyl's inequality and \eqref{eq:G_gap}, we conclude that \eqref{eq:singularity_b} holds for every $G^{(j)}$ for $0\le j\le J.$ Observe that for our particular choice of $B^*,$ we have $\lambda_1(B^*)\asymp K$ and $\lambda_k(B^*)=c'\beta_n$ for all $2\le k\le K.$ And,
$$|\lambda_K(B^*G^{(0)})|\ge\frac{1}{\lambda_1((B^*G^{(0)})^{-1})}\ge\frac{c'\beta_n}{c_1}.$$
We can take $c'=c_1$ so that $B^*G^{(0)}$ and similarly all the $B^*G^{(j)}$'s satisfy the second inequality in \eqref{eq:spectral_gap}. It is also easy to see that the first inequality in \eqref{eq:spectral_gap} holds for $B^*G^{(0)}$ for some appropriate choice of $c_2.$ Next, using \eqref{eq:G_gap} and \eqref{eq:perturb_bound} we have
\begin{equation}
    \|B^*G^{(j)}-B^*G^{(0)}\|\le\|B^*\|\cdot\|G^{(j)}-G^{(0)}\|\le K(2c_1^{1/2}c_0\mathrm{err}_n+c_0^2\mathrm{err}_n^2)\le C_1K\mathrm{err}_n
    \label{eq:b_perturb_norm}
\end{equation}
for some constant $C_1>0$ as $\mathrm{err}_n\to0.$ Thus, by Weyl's inequality,
$$\lambda_1(B^*G^{(j)})-|\lambda_2(B^*G^{(j)})|\ge(\lambda_1(B^*G^{(0)})-C_1K\mathrm{err}_n)-(|\lambda_2(B^*G^{(j)})|+C_1K\mathrm{err}_n)\ge C_2K$$
for some constant $C_2>0.$ Hence, the first inequality in \eqref{eq:spectral_gap} holds true for all $B^*G^{(j)}$'s. Also, \eqref{eq:perron_root} and \eqref{eq:non_negligible_presence} are direct consequences of \eqref{eq:deg_condn} for every $\Pi^{(j)}$ with $0\le j\le J.$
\paragraph{Separation of the constructed configurations}
We shall show that the above configuration is \emph{separated}(up to absolute constants) by
\begin{align}
\label{eq:temp_placeholder_err_lower_bd}
\int_0^\infty\min\left\{\frac{\mathrm{err}_n}{t\wedge1},\,1\right\}\text{d}F_n(t)
\end{align}
with respect to the loss function $\mathcal{L}(\cdot,\cdot)$, yet the \emph{distributions of the privatized output} matrices $\pM(\Pi^{(j)})$ remain \emph{close} to that of $\pM(\Pi^{(0)})$ in terms of the Kullback--Leibler divergence, regardless of the specific privatization mechanism used to generate $\pM(\Pi^{(0)})$. Consequently, based on the privatized data alone, it is impossible to reliably distinguish whether the underlying (non-private) data originated from $\Pi^{(0)}$ or from $\Pi^{(j)}$. This implies that \emph{accurate reconstruction} of the true membership profile $\Pi$ is infeasible for any algorithm unless the separation between feasible membership profiles exceeds the \emph{fundamental threshold} stated in \eqref{eq:temp_placeholder_err_lower_bd} and any estimator constructed exclusively using any privatized network data satisfying the edge $\varepsilon$-local differential privacy constraint will suffer the error described in the theorem irrespective of the estimator used to reconstruct $\Pi$ and the mechanism adopted to privatize the data.

Observe that any $0\le j<k\le J,$
\begin{align}
    \mathcal{L}(\Pi^{(j)},\,\Pi^{(k)})&=\frac{1}{n}\sum_{i=1}^n\|\Pi_{i*}^{(j)}-\Pi_{i*}^{(k)}\|_1=\frac{\gamma_n}{n}\sum_{i=1}^{n_0}\frac{1}{\theta_i \wedge \widetilde \theta}\|(\Gamma^{(j)}-\Gamma^{(k)})_{i*}\|_1.
    \label{eq:par_sep}
\end{align}
We know that $\|\widetilde\Gamma^{(j)}-\widetilde\Gamma^{(k)}\|_1\ge\lfloor n_0/2\rfloor\times\lfloor K/2\rfloor/8.$ Thus, at least $\lfloor\lfloor n_0/2\rfloor/8\rfloor$ rows of $\widetilde\Gamma^{(j)}-\widetilde\Gamma^{(k)}$ should have positive Hamming distance. And due to the construction of $\Gamma^{(j)}$ from $\widetilde\Gamma^{(j)}$ in \eqref{eq:construct_gamma}, we have
\begin{align}
&0\le\|(\Gamma^{(j)}-\Gamma^{(k)})_{i*}\|_1\le2\lfloor K/2\rfloor\nonumber\\
\text{and}\quad&\sum_{i=1}^{n_0}\|(\Gamma^{(j)}-\Gamma^{(k)})_{i*}\|_1=4\|\widetilde\Gamma^{(j)}-\widetilde\Gamma^{(k)}\|_1\ge\lfloor n_0/2\rfloor\times\lfloor K/2\rfloor/2.
\label{eq:constraint}
\end{align}
Treating the $L_1$ norms of the rows as weights, the RHS of \eqref{eq:par_sep} would be minimized when the large $\theta_i$'s correspond to large $L_1$ norms $\|(\Gamma^{(j)}-\Gamma^{(k)})_{i*}\|_1,$ subject to the constraint in \eqref{eq:constraint} and $n_1:=2\lfloor\lfloor n_0/2\rfloor/8\rfloor$ of the $L_1$ norms being positive. If we assume without loss of generality that $\theta_{(1)}\le\theta_{(2)}\le\ldots\le\theta_{(n_0)},$ then the RHS of \eqref{eq:par_sep} is minimized as
\begin{align}
    \mathcal{L}(\Pi^{(j)},\,\Pi^{(k)})&\ge\frac{\gamma_n}{n}\sum_{i=n_0-n_1+1}^{n_0}\frac{2\lfloor K/2\rfloor}{\theta_i\wedge\widetilde\theta}\gtrsim\frac{\mathrm{err}_n}{n}\sum_{i=n_0-n_1+1}^{n_0}\frac{1}{\nu_i\wedge1},
    \label{eq:par_sep_1}
\end{align}
recalling the definition of $\gamma_n$ in \eqref{eq:gamma_n}. Now one may observe that for $1\le i\le n_0,$ we have
$$\frac{\mathrm{err}_n}{n(\nu_i\wedge 1)}\le\frac{\mathrm{err}_n}{n(\mathrm{err}_n\wedge 1)}\le\frac{1}{n}\lesssim\mathrm{err}_n$$
so that we can add finitely many terms in the RHS of \eqref{eq:par_sep_1} without changing the inequality asymptotically. Thus, continuing from \eqref{eq:par_sep_1}, we obtain
\begin{align}
    \mathcal{L}(\Pi^{(j)},\,\Pi^{(k)})&\gtrsim\frac{\mathrm{err}_n}{n}\sum_{i=n_0-n_1-1}^{n_0}\frac{1}{\nu_i\wedge1}\ge\mathrm{err}_n\int_{\tau_n}^{c_n}\frac{\text{d}F_n(t)}{t\wedge 1}.
    \label{eq:par_sep_2}
\end{align}
where $\tau_n$ is the quantile $F_n^{-1}((n_0-n_1-1)/n)$ of $F_n,$ the empirical distribution function of $\{\nu_1,\,\nu_2,\,\ldots,\,\nu_n\}.$ A simple calculation shows that $n_1=2\lfloor\lfloor n_0/2\rfloor/8\rfloor\ge n_0/8-2.$ Thus,
\begin{equation}
\int_{\tau_n}^{c_n}\text{d}F_n(t)=\frac{n_1+2}{n}\ge\frac{n_0}{8n}=\frac{1}{8}\int_{\mathrm{err}_n}^{c_n}\text{d}F_n(t).
\label{eq:tau_n}
\end{equation}
Now again continuing from \eqref{eq:par_sep_2},
\begin{align}
    \mathcal{L}(\Pi^{(j)},\,\Pi^{(k)})&\gtrsim\mathrm{err}_n\int_{\tau_n}^{c_n}\frac{\text{d}F_n(t)}{t\wedge 1}\ge\frac{\mathrm{err}_n}{c_n\wedge1}\int_{\tau_n}^{c_n}\text{d}F_n(t)\nonumber\\
    &\ge\frac{\mathrm{err}_n\rho}{8(\rho c_n\wedge1)}\int_{\mathrm{err}_n}^{c_n}\text{d}F_n(t)&[\text{From Definition \ref{def:constr_deg_cor} and \eqref{eq:tau_n}}]\nonumber\\
    &\ge\frac{\mathrm{err}_n\rho}{8(\rho c_n\wedge 1)}\int_{\rho c_n}^{c_n}\text{d}F_n(t)&[\text{as}~\mathrm{err}_n\le\rho c_n]\nonumber\\
    &\ge\frac{\mathrm{err}_n\rho}{8}\int_{\rho c_n}^{c_n}\frac{\text{d}F_n(t)}{t\wedge 1}\nonumber\\
    &\ge\frac{\mathrm{err}_n\rho a_0}{8}\int_{\mathrm{err}_n}^\infty\frac{\text{d}F_n(t)}{t\wedge 1}&[\text{Due to \eqref{eq:defn2.2}}]\nonumber\\
    &\gtrsim\int_{\mathrm{err}_n}^\infty\min\left\{\frac{\mathrm{err}_n}{t\wedge1},1\right\}\text{d}F_n(t)\nonumber\\
    &\gtrsim\int_{0}^\infty\min\left\{\frac{\mathrm{err}_n}{t\wedge1},1\right\}\text{d}F_n(t).&[\text{Due to \eqref{eq:F_n-condn}}]
    \label{eq:separation}
\end{align}
So the parameters $\Pi^{(j)}$ are well separated. 

\paragraph{Controlling the divergence between the induced distributions}
Recall $\mathcal{R}_\varepsilon$, the collection of randomized mechanisms ensuring $\varepsilon$-edge local differential privacy for adjacency matrices. Using Corollary 1 of \cite{6686179}, we obtain:
\begin{equation}
\label{eq:unif_upperbd_kl}
\sup_{\mathcal{M} \in \mathcal{R}_\varepsilon} \KL(\mathcal{M}(A_{\Pi^{(j)}}) \| \mathcal{M}(A_{\Pi^{(0)}})) 
\leq 4(e^\varepsilon - 1)^2 \sum_{e \in \binom{n}{2}} \|P_{\Pi^{(j)},\,e} - P_{\Pi^{(0)},\,e}\|_{\mathrm{TV}}^2,
\end{equation}
where, for $j \in [J]$, $\mathcal{M}(A_{\Pi^{(j)}})$ denotes the distribution of the adjacency matrix $A_{\Pi^{(j)}}$ generated from the model in \eqref{eq:adj_matrix} using parameter $\Pi^{(j)}$, after applying the randomized mechanism $\mathcal{M}$ to ensure $\varepsilon$-edge LDP, and $P_{\Pi^{(j)},\,e}:=\theta_k\theta_\ell(\Pi^{(j)}_{k*})^\top B^*\Pi^{(j)}_{\ell*}$ for $e=(k,\ell)$. If 
\begin{align}
\label{eq:bound_kl_fano}
\sup_{\mathcal{M} \in \mathcal{R}_\varepsilon} \frac{1}{J} \sum_{j=1}^{J} \KL(\mathcal{M}(A_{\Pi^{(j)}}) \| \mathcal{M}(A_{\Pi^{(0)}})) \leq \alpha \log J,
\end{align}
for some $\alpha \in (0,1/8)$, then, using the arguments in Proposition 2.3 and Theorem 2.5 of \cite{MR2724359}, we can show that
\begin{align}
\label{eq:fano_lowe_bound}
&\inf_{\widehat{\Pi} \in \mathfrak{M}_\varepsilon} \sup_{(\Pi, B) \in \mathcal{Q}_n(\theta)} \P\left[\mathcal{L}(\widehat{\Pi}, \Pi) \gtrsim\int_{0}^\infty\min\left\{\frac{\mathrm{err}_n}{t\wedge1},1\right\}\text{d}F_n(t)\right] \nonumber\\
&\hskip 20em \geq \frac{\sqrt{J}}{1+\sqrt{J}}\left(1 - 2\alpha - \sqrt{\frac{2\alpha}{\log J}}\right).
\end{align}
By taking $J \rightarrow \infty$, we can ensure that the above probability is greater than $1-o(n^{-2})$. Therefore, it is enough to show \eqref{eq:unif_upperbd_kl}.

Observe that, for all $\ell \in [J]$, using \eqref{eq:unif_upperbd_kl}, we have
\begin{align}
\sup_{\mathcal{M} \in \mathcal{R}_\varepsilon}
    \KL(\mathcal{M}(A_{\Pi^{(\ell)}})\|\mathcal{M}(A_{\Pi^{(0)}}))&\le4(e^\varepsilon-1)^2\sum_{e \in {n \choose 2}}\|P_{\Pi^{(\ell)},\,e}-P_{\Pi^{(0)},\,e}\|_{\mathrm{TV}}^2\notag\\
    &=4(e^\varepsilon-1)^2\sum_{1\le i<j\le n}\left(\Omega_{ij}^{(\ell)}-\Omega_{ij}^{(0)}\right)^2
    \label{eq:kl}
\end{align}
where $\Omega_{ij}^{(\ell)}=\theta_i\theta_j{\Pi_{i*}^{(\ell)}}^\top B^*\Pi_{j*}^{(\ell)}$ for $1\le i<j\le n$ and $0\le \ell\le J.$ Observe that for $n_0<i<j\le n,$ we have $\Omega_{ij}^{(\ell)}=\Omega_{ij}^{(0)}.$ So only contributing terms will correspond to $(i,j)$ pairs where either $0<i<j\le n_0$ and $0<i\le n_0<j\le n.$ Let us consider these two cases separately.

If $0<i<j\le n_0,$ by the construction of $\Gamma^{(\ell)}$, we have
\begin{align*}
    \Omega_{ij}^{(0)}&=\theta_i\theta_j\left(\frac{1}{K}1_K\right)^\top B^*\left(\frac{1}{K}1_K\right),\\
    \Omega_{ij}^{(\ell)}&=\theta_i\theta_j\left(\frac{1}{K}1_K+\frac{\gamma_n}{\theta_i\wedge \widetilde \theta}\Gamma_{i*}^{(\ell)}\right)^\top B^*\left(\frac{1}{K}1_K+\frac{\gamma_n}{\theta_j\wedge \widetilde \theta}\Gamma_{j*}^{(\ell)}\right)\\
    &=\Omega_{ij}^{(0)}+\frac{\theta_i\theta_j\gamma_n^2\beta_n}{(\theta_i\wedge \widetilde \theta)(\theta_j\wedge \widetilde \theta)}{\Gamma_{i*}^{(\ell)}}^\top\Gamma_{j*}^{(\ell)}.
\end{align*}
Now observe that for every $1\le i\le n_0,$ using the construction of $c_n$ in Definition \ref{def:constr_deg_cor} and \eqref{eq:F_n-condn}, we can conclude that $\theta_i/\widetilde{\theta}<c_n\le\widetilde C$ for some $\widetilde C>0.$ Thus, for every $1\le i\le n,$
\begin{equation}
    \frac{\theta_i}{\theta_i\wedge\widetilde\theta}\le\max\left\{1,\,\frac{\theta_i}{\widetilde{\theta}}\right\}\le\max\left\{1,\,\widetilde C\right\}.
    \label{eq:theta_bound}
\end{equation}
Using \eqref{eq:gamma_n} we have for $1\le i<j\le n_0,$
\begin{align*}
    \left(\Omega_{ij}^{(\ell)}-\Omega_{ij}^{(0)}\right)^2\le \gamma_n^4c'^2\beta_n^2\left({\Gamma_{i*}^{(\ell)}}^\top\Gamma_{j*}^{(\ell)}\right)^2&\le\max\{1,\,\widetilde C\}^2K^2\gamma_n^4c^2_0\beta_n^2\\
    & \lesssim \frac{K^4\max\{1,\,\widetilde C\}^2c^2_0}{(1-2p_\varepsilon)^4\beta^2_nn^2\widetilde\theta^4}.
\end{align*}
Since $\mathrm{err}_n \ll 1$ by assumption, therefore we have
\[
\frac{K^3}{(1-2p_\varepsilon)^2\beta^2_nn\tilde\theta^4} \ll 1.
\]
This in turn implies
\[
(e^\varepsilon-1)^2\sum_{1 \le i,j \le n_0}\left(\Omega_{ij}^{(\ell)}-\Omega_{ij}^{(0)}\right)^2 \lesssim_{\varepsilon_0} Kn_0\max\{1,\,\widetilde C\}^2c^2_0\lesssim c^2_0Kn_0,
\]
using $\varepsilon \in [0,\varepsilon_0]$. Therefore
\begin{align}
\label{eq:kl_bd_jn_1}
(e^\varepsilon-1)^2\sum_{e=(i,j):n_0\le i,j \le n}\|P_{\Pi^{(\ell)},e}-P_{\Pi^{(0)},e}\|^2_{\mathrm{TV}} \lesssim c^2_0Kn_0.
\end{align}

Next, for $0<i\le n_0<j\le n,$ suppose $j\in\mathcal{C}_{k},$ that is, node $j$ belongs to cluster $\mathcal{C}_{k}.$ Then we have
\begin{align*}
    \Omega_{ij}^{(0)}&=\theta_i\theta_j\left(\frac{1}{K}1_K\right)^\top B^*e_{k},\\
    \Omega_{ij}^{(\ell)}&=\theta_i\theta_j\left(\frac{1}{K}1_K+\frac{\gamma_n}{\theta_i\wedge \widetilde \theta}\Gamma_{i*}^{(\ell)}\right)^\top B^*e_{k}=\Omega_{ij}^{(0)}+\frac{\theta_i\theta_j\gamma_n\beta_n}{\theta_i\wedge \widetilde \theta}{\Gamma_{i*}^{(\ell)}}^\top e_{k}.
\end{align*}
So in this case, again due to \eqref{eq:theta_bound}, we have
\begin{align*}
    \left(\Omega_{ij}^{(\ell)}-\Omega_{ij}^{(0)}\right)^2=\frac{\theta_i^2\theta_j^2\gamma_n^2\beta_n^2}{(\theta_i \wedge \widetilde \theta)^2}\left({\Gamma_{i*}^{(\ell)}}^\top e_{k}\right)^2&\le\max\{1,\,\widetilde C\}\theta_j^2\gamma_n^2\beta_n^2\\
    & \le \max\{1,\,\widetilde C\}c^2_0\times\frac{\theta^2_jK}{n\widetilde\theta^2(1-2p_\varepsilon)^2}.
\end{align*}
Therefore, we have
\begin{align}
    (e^\varepsilon-1)^2\sum_{1 \le i<n_0\le j \le n}\left(\Omega_{ij}^{(\ell)}-\Omega_{ij}^{(0)}\right)^2 & \le \max\{1,\,\widetilde C\}c^2_0\times\frac{n_0K}{(1+e^\varepsilon)^2}\nonumber\\
    & \lesssim c^2_0Kn_0.\nonumber
\end{align}
This implies
\begin{align}
\label{eq:kl_bd_jn_2}
(e^\varepsilon-1)^2\sum_{e=(i,j):1 \le i<n_0\le j \le n}\|P_{\Pi^{(\ell)},e}-P_{\Pi^{(0)},e}\|^2_{\mathrm{TV}} \lesssim c^2_0Kn_0.
\end{align}
Plugging \eqref{eq:kl_bd_jn_1} and \eqref{eq:kl_bd_jn_2} in \eqref{eq:kl}, we can conclude that
\begin{align}
   \sup_{\mathcal{M} \in \mathcal{R}_\varepsilon} \KL(\mathcal{M}(A_{\Pi^{(\ell)}}) \| \mathcal{M}(A_{\Pi^{(0)}})) \lesssim c^2_0Kn_0.\nonumber
\end{align}
Using Jensen's inequality, this implies
\begin{align}
    \sup_{\mathcal{M} \in \mathcal{R}_\varepsilon}\frac{1}{J}\sum_{\ell=1}^{J} \KL(\mathcal{M}(A_{\Pi^{(\ell)}}) \| \mathcal{M}(A_{\Pi^{(0)}})) \lesssim c^2_0n_0K.
\end{align}
By construction, $\log J=\frac{s}{8}\log2\asymp mr\asymp n_0K$. This implies that we can choose $c_0$ small enough to ensure
\[
\sup_{\mathcal{M} \in \mathcal{R}_\varepsilon}\frac{1}{J}\sum_{j=1}^{J} \KL(\mathcal{M}(A_{\Pi^{(j)}}) \| \mathcal{M}(A_{\Pi^{(0)}})) \lesssim c_0\log J \lesssim \alpha\log J,
\]
for some constant $0<\alpha<\frac{1}{8}$. 
%Using $\mathrm{err}_n \rightarrow 0$, as $n \rightarrow \infty$, we can conclude using \eqref{eq:F_n-condn} that $n_0$ grows with $n$. 
Therefore, we can take $J \rightarrow \infty$, to conclude \eqref{eq:fano_lowe_bound}.

\subsection{Proof of Theorem \ref{thm:lower_bound_informal} without (\ref{eq:F_n-condn})}
When \eqref{eq:F_n-condn} does not hold, we shall construct a different least favorable configuration. Let us recall that in this regime
\begin{align}
\liminf_{n \rightarrow \infty}\frac{n\left|\{i:\nu_i \le \mathrm{err}_n\}\right|}{\sum_{i:~\nu_i \ge \mathrm{err}_n}\frac{\mathrm{err}_n}{\nu_i \wedge 1}} =  \infty,
    \label{eq:comp_F_n-condn}
\end{align}
where $\nu_i=\theta_i \wedge \widetilde\theta$. Let us arrange the $\theta_i$'s in increasing order and consider
\begin{align}
\label{eq:n_0_case_2}
n_0=\max\left\{1\le i\le n:\frac{\theta_{(i)}}{\widetilde\theta}\le\mathrm{err}_n\right\}.
\end{align}
As in the case where \eqref{eq:F_n-condn}, in our constructed configuration, the last $n-n_0$ nodes will be pure nodes equally distributed among the communities such that \eqref{eq:deg_condn} still holds. Let us construct $\Pi^{(0)}$ as before using $n_0$ defined in \eqref{eq:n_0_case_2}. We construct $\Gamma^{(0)},\,\Gamma^{(1)},\,\ldots,\,\Gamma^{(J)}$ as in \eqref{eq:construct_gamma} (using Gilbert-Varshamov theorem) and choose the same $B^*.$ Next, let us consider
$$\Pi^{(j)}=\Pi^{(0)}+\frac{c_0}{K}\Gamma_{(j)},\quad0\le j\le J.$$
Also define $G^{(j)}$ for $0\le j\le J$ as in the previous regime. Observe that, since $\Pi^{(0)}$ and $B^*$ are unchanged, the regularity conditions \eqref{eq:singularity_b}, \eqref{eq:spectral_gap}, \eqref{eq:perron_root} and \eqref{eq:non_negligible_presence} hold for $G^{(0)}$ and $B^*G^{(0)}.$ 
Now note that,
\begin{align}
    \|G^{(j)}-G^{(0)}\|&\le 2\|G^{(0)}\|^{1/2}\|c_0^2K^{-2}{\Gamma^{(j)}}^\top\Theta^2\Gamma^{(j)}\|^{1/2}+\|c_0^2K^{-2}{\Gamma^{(j)}}^\top\Theta^2\Gamma^{(j)}\|.
    \label{eq:perturb_norm_2}
\end{align}
As before, one can derive
$$\|c_0^2K^{-2}{\Gamma^{(j)}}^\top\Theta^2\Gamma^{(j)}\|\lesssim\frac{c_0}{K}.$$
So one can choose $c_0$ small enough such that the regularity conditions \eqref{eq:singularity_b}, \eqref{eq:spectral_gap}, \eqref{eq:perron_root} and \eqref{eq:non_negligible_presence} continue to hold for $G^{(j)}$ and $B^*G^{(j)}$ for all $0\le j\le J.$
\paragraph{Separation of the constructed configurations}
Observe that
\begin{align}
\label{eq:no_a_1_lb}
    \mathcal{L}(\Pi^{(j)},\,\Pi^{(k)})&=\frac{1}{n}\sum_{i=1}^n\|\Pi_{i*}^{(j)}-\Pi_{i*}^{(k)}\|_1=\frac{c_0}{nK}\sum_{i=1}^{n_0}\|(\Gamma^{(j)}-\Gamma^{(k)})_{i*}\|_1\nonumber\\
    &=\frac{2c_0(K-1)n_0}{nK}\ge\frac{c_0n_0}{8n}\ge\frac{c_0}{8n}\sum_{i:~\nu_i \ge \mathrm{err}_n}\frac{\mathrm{err}_n}{\nu_i \wedge 1}\\
    &\ge \frac{c_0}{16}\int_{0}^\infty\min\left\{\frac{\mathrm{err}_n}{t\wedge1},1\right\}\text{d}F_n(t),
\end{align}
where the second last inequality follows using the definition of the regime and the last inequality follows using \eqref{eq:comp_F_n-condn}.
This shows that the newly constructed parameters $\Pi^{(j)}$ are well separated in terms of $\mathcal L(\cdot,\cdot)$ norm.

\paragraph{Controlling the divergence of the induced distributions}
By construction, we have for $0<i<j\le n_0,$ 
\begin{align*}
    \Omega_{ij}^{(0)}&=\theta_i\theta_j\left(\frac{1}{K}1_K\right)^\top B^*\left(\frac{1}{K}1_K\right)=\theta_i\theta_j\left[1-\left(1-\frac{1}{K}\right)c'\beta_n\right],\\
    \Omega_{ij}^{(\ell)}&=\theta_i\theta_j\left(\frac{1}{K}1_K+\frac{c_0}{K}\Gamma_{i*}^{(\ell)}\right)^\top B^*\left(\frac{1}{K}1_K+\frac{c_0}{K}\Gamma_{j*}^{(\ell)}\right)\\
    &=\Omega_{ij}^{(0)}+\frac{\theta_i\theta_jc_0^2}{K^2}{\Gamma_{i*}^{(\ell)}}^\top B^*\Gamma_{j*}^{(\ell)},
\end{align*}
which means
$$\left(\Omega_{ij}^{(\ell)}-\Omega_{ij}^{(0)}\right)^2\le\frac{\theta_i^2\theta_j^2c_0^4c'^2\beta_n^2}{K^4}\left({\Gamma_{i*}^{(\ell)}}^\top\Gamma_{j*}^{(\ell)}\right)^2\le\frac{\theta_i^2\theta_j^2c_0^4c'^2\beta_n^2}{K^2}\le\frac{K\theta_j^2c_0^4c'^2}{(1-2p_\varepsilon)^2n\widetilde\theta^2}$$
since all the $\theta_i$'s satisfy
\begin{equation}
    \frac{\theta_i}{\widetilde\theta}\le\frac{K^{3/2}}{(1-2p_\varepsilon)\beta_n\sqrt{n}\widetilde\theta^2},\quad\text{for }1\le i\le n_0.
    \label{eq:theta_new_bound}
\end{equation}
This implies
\begin{align}
    \sum_{0<i<j\le n_0}\left(\Omega_{ij}^{(
    \ell)}-\Omega_{ij}^{(0)}\right)^2&\le\frac{c_0^4c'^2Kn_0}{(1-2p_\varepsilon)^2}\le\frac{c_0^4c'^2}{(1-2p_\varepsilon)^2}n_0K.
    \label{eq:case2.1}
\end{align}

Next, for $0<i\le n_0<j\le n,$ suppose the node $j$ belongs to cluster $\mathcal C_k.$ Then we have
\begin{align*}
    \Omega_{ij}^{(0)}&=\theta_i\theta_j\left(\frac{1}{K}1_K\right)^\top B^*e_{k},\\
    \Omega_{ij}^{(\ell)}&=\theta_i\theta_j\left(\frac{1}{K}1_K+\frac{c_0}{K}\Gamma_{i*}^{(\ell)}\right)^\top B^*e_{k}=\Omega_{ij}^{(0)}+\frac{\theta_i\theta_jc_0c'\beta_n}{K}{\Gamma_{i*}^{(\ell)}}^\top e_{k},
\end{align*}
which means
$$\left(\Omega_{ij}^{(\ell)}-\Omega_{ij}^{(0)}\right)^2\le\frac{\theta_i^2\theta_j^2c_0^2c'^2\beta_n^2}{K^2}\left({\Gamma_{i*}^{(\ell)}}^\top e_k\right)^2\le\frac{\theta_i^2\theta_j^2c_0^2c'^2\beta_n^2}{K^2}\lesssim\frac{K\theta_j^2c_0^4c'^2}{(1-2p_\varepsilon)^2n\widetilde\theta^2},$$
again due to \eqref{eq:theta_new_bound} and the uniform bound of all $\theta_i$'s. Thus, we have
\begin{align}
    \sum_{0<i\le n_0<j\le n}\left(\Omega_{ij}^{(\ell)}-\Omega_{ij}^{(0)}\right)^2\le\frac{c_0^4c'^2}{(1-2p_\varepsilon)^2}nK.
    \label{eq:case2.2}
\end{align}

Plugging in the bounds in \eqref{eq:case2.1} and \eqref{eq:case2.2} in \eqref{eq:kl}, using \eqref{eq:unif_upperbd_kl} we get 
\begin{align*}
&\sup_{\mathcal{M} \in \mathcal{R}_\varepsilon}\KL(\mathcal{M}(A_{\Pi^{(\ell)}})\|\mathcal{M}(A_{\Pi^{(0)}}))\\
&\le4(e^\varepsilon-1)^2\left[\sum_{0<i<j\le n_0}\left(\Omega_{ij}^{(l)}-\Omega_{ij}^{(0)}\right)^2+\sum_{0<i\le n_0<j\le n}\left(\Omega_{ij}^{(l)}-\Omega_{ij}^{(0)}\right)^2\right]\\
    &\le c^2_0c'^2(1\vee \varepsilon_0)^2nK.
\end{align*}
In the above display, the last line follows from the fact that $\varepsilon \in [0,\varepsilon_0]$. Now, we can use the construction of the least favorable configuration and take $c_0$ small enough to ensure
$$\sup_{\mathcal{M} \in \mathcal{R}_\varepsilon}\frac{1}{J}\sum_{\ell=1}^J\KL(\mathcal{M}(A_{\Pi^{(\ell)}})\|\mathcal{M}(A_{\Pi^{(0)}})) \le\alpha \log J,$$
for some constant $0<\alpha<\frac{1}{8}$. This implies \eqref{eq:fano_lowe_bound} in this subcase using \eqref{eq:bound_kl_fano} and \eqref{eq:no_a_1_lb}.

\section{Analysis of the estimation risk of {\sf PriME} Algorithm}
\subsection{Proof of Theorem \ref{thm:l_infty_pertb}}  
To establish Theorem \ref{thm:l_infty_pertb}, we adapt the leave-one-out technique developed in \citet{ke2022optimal} to control eigenvector perturbations.  
Let $\pW = \pM - \mathbb{E}[\pM]$ for {$\pM$ defined in (3.2) of the main text}.  
Observe that  
\[
\pM = \Omega - \mathsf{diag}(\Omega) + \pW.
\]  
For each $i \in [n]$, define the \emph{leave-one-out residual matrix} $\pW^{(i)} \in \mathbb{R}^{n \times n}$ by  
\begin{equation}
\label{Eq:leave_one_out_w}
(\pW^{(i)})_{k\ell} =
\begin{cases}
(\pW)_{k\ell}, & \text{if } k \neq i \text{ and } \ell \neq i, \\
0, & \text{otherwise.}
\end{cases}
\end{equation}
Correspondingly, we define the \emph{leave-one-out matrix}
\[
\pM^{(i)} = \Omega - \mathsf{diag}(\Omega) + \pW^{(i)}.
\]
Let $\widehat{\lambda}_1, \ldots, \widehat{\lambda}_K$ denote the $K$ largest eigenvalues of $\pM$, and let $\widehat{\Xi}$ denote the corresponding matrix of eigenvectors.  
Similarly, denote by $(\widehat{\lambda}^{(i)}, \widehat\Xi^{(i)})$ the leading $K$ eigenvalues and eigenvectors of the leave-one-out matrix $\pM^{(i)}$.

\subsubsection{Eigenstructure of the Population Mean}
Let us consider the population mean matrix 
\[
\mathbb E[\pM]=\Omega-\mathsf{diag}(\Omega).
\]
We can show the following property of the eigenvalues $\Omega$ using Lemma C.2 of \citet{JIN2023}. 

\begin{lem}[Lemma C.2, \citet{JIN2023}]
\label{lem:pop_eigval}
Under Assumptions \eqref{eq:singularity_b}, \eqref{eq:spectral_gap}, \eqref{eq:perron_root}, and \eqref{eq:non_negligible_presence}, the eigenvalues of $\Omega$ satisfy the following:
\begin{align}
    & K^{-1}\|\theta\|^2 \lesssim \lambda_1 \lesssim \|\theta\|^2. \ \text{Furthermore, if } \beta_n = o(1), \text{ then } \lambda_1 \asymp \|\theta\|^2, \label{eq:lambda1_bounds} \\
    & \lambda_1 - |\lambda_2| \asymp \lambda_1, \label{eq:spectral_gap_omega} \\
    & \lambda_k \asymp \beta_n K^{-1}\|\theta\|^2, \quad \text{for all } 2 \le k \le K. \label{eq:lambda_k_bounds}
\end{align}
\end{lem}
Next, we examine the eigenstructure of $\Omega$. The behavior of its eigenvectors is characterized by the following lemma.

\begin{lem}[Lemma C.3, \citet{JIN2023}]
\label{lem:true_eig_vect}
Under Assumptions \eqref{eq:singularity_b}, \eqref{eq:spectral_gap}, \eqref{eq:perron_root}, and \eqref{eq:non_negligible_presence}, let $\Xi \in \mathbb{R}^{n \times K}$ denote the eigenvector matrix of $\Omega$.  
If the sign of the leading eigenvector $\Xi_{*1}$ is chosen such that $\sum_{i=1}^n \Xi_{i1} > 0$, then all entries of $\Xi_{*1}$ are positive and satisfy $\Xi_{i1} \asymp \theta_i / \|\theta\|$ for all $i \in [n]$.  
Furthermore, for each $i \in [n]$, the remaining coordinates satisfy 
\[
\|\Xi_{i,2:K}\| \lesssim \sqrt{K} \, \frac{\theta_i}{\|\theta\|}.
\]
\end{lem}

\subsection{Perturbation of the sample eigenvectors}
To prove the perturbation bounds for the sample eigenvectors of $\pM$, let us consider the following lemma characterizing the deviation of $\pM$ from $\Omega$ in terms of the spectral norm.

\begin{lem}
\label{lem:dev_h_0}
Assume \eqref{eq:singularity_b}, \eqref{eq:spectral_gap}, \eqref{eq:perron_root}, and \eqref{eq:non_negligible_presence}. 
Suppose further that $\theta_i = O(1)$ for all $i \in [n]$ and $\varepsilon \in [0, \varepsilon_0]$. 
If
\begin{align}
\label{eq:imp_condn}
    \frac{K^3 (\log n)^2}{n \tilde{\theta}^4\beta_n^2 (1 - 2p_\varepsilon)^2} \ll 1
\end{align}
for all $n \ge n_0$ and some fixed $n_0 > 0$, where $\tilde{\theta}=\|\theta\|_2/\sqrt{n}$,
then
\[
\|\pM - \Omega\| \lesssim \frac{\sqrt{n}}{1 - 2p_\varepsilon},
\]
with probability at least $1 - o(n^{-3})$.
\end{lem}

\begin{proof}
We begin by decomposing
\[
\|\pM - \Omega\| \le \|\pM - \Omega - \mathsf{diag}(\Omega)\| + \|\mathsf{diag}(\Omega)\|.
\]
Since $\theta_i = O(1)$ for all $i \in [n]$, it follows that 
\[
\|\mathsf{diag}(\Omega)\| \lesssim 1.
\]
Next, by definition of $\pW = \pM - \mathbb{E}[\pM]$, we have
\[
\mathbb{E}\!\left[(\pW)_{ij})^2\right] 
    \lesssim \frac{1}{(1 - 2p_\varepsilon)^2}
    \left\{\frac{1}{1 + e^{\varepsilon}} + \theta_i \theta_j\right\} 
    \lesssim \frac{1}{(1 - 2p_\varepsilon)^2}.
\]
Consequently,
\[
\tilde{\sigma}^2 
    = \max_{i} \sum_{j \neq i} \mathbb{E}\!\left[(\pW)_{ij})^2\right] 
    \lesssim \frac{n}{(1 - 2p_\varepsilon)^2},
    \quad\text{and}\quad
    \sigma_\star 
    = \max_{i,j} |(\pW)_{ij}| 
    \lesssim \frac{1}{1 - 2p_\varepsilon}.
\]
Applying Corollary 3.12 of \citet{10.1214/15-AOP1025}, we obtain
\[
\mathbb{P}\!\left[\|\pW\| \ge C \tilde{\sigma} + t_\star\right] 
    \le n \exp\!\left(-\frac{t_\star^2}{c \sigma_\star^2}\right).
\]
Choosing $t_\star = c' \sqrt{\log n}\,\sigma_\star$ and invoking a union bound argument yield that, with probability at least $1 - o(n^{-3})$,
\[
\|\pM - \Omega - \mathsf{diag}(\Omega)\| 
    \lesssim \frac{\sqrt{n}}{1 - 2p_\varepsilon}.
\]
Combining the above bounds completes the proof.
\end{proof}
Next, we consider the $L_2$ perturbation of the empirical eigenvectors of $\pM$ around those of $\Omega$. 

\begin{lem}
\label{lem:l_2_perturb_eig_vect}
Let $\widehat{\Xi} \in \mathbb{R}^{n \times K}$ denote the matrix of the top $K$ eigenvectors of $\pM$, and let $\Xi$ denote the corresponding eigenvector matrix of $\Omega$. 
Suppose that $\theta_i = O(1)$ for all $i \in [n]$ and $\varepsilon \in [0, \varepsilon_0]$. 
Under Assumptions \eqref{eq:singularity_b}, \eqref{eq:spectral_gap}, \eqref{eq:perron_root}, \eqref{eq:non_negligible_presence}, and condition \eqref{eq:imp_condn}, there exist $\omega \in \{-1, 1\}$ and an orthogonal matrix $O \in \mathbb{R}^{(K-1) \times (K-1)}$ such that
\[
\|\Xi_{*1} - \omega \widehat{\Xi}_{*1}\|_2 
    \lesssim \frac{K\sqrt{n}}{(1 - 2p_\varepsilon)\|\theta\|^2},
\]
and
\[
\|\Xi_{*,2:K} - \widehat{\Xi}_{*,2:K}O\|_F 
    \lesssim \frac{K^{3/2}\sqrt{n}}{(1 - 2p_\varepsilon)\beta_n\|\theta\|^2},
\]
with probability at least $1 - o(n^{-3})$.
\end{lem}

\begin{proof}
By Weyl’s inequality \citep{bhatia1997matrix}, 
\begin{align}
\label{eq:prtb_eig_val}
|\widehat{\lambda}_k - \lambda_k| 
    \le \|\pM - \Omega\| 
    \lesssim \frac{\sqrt{n}}{1 - 2p_\varepsilon},
\end{align}
where the last inequality follows from Lemma~\ref{lem:dev_h_0} and the assumption $\theta_i = O(1)$ for all $i \in [n]$.

Using Lemmas~\ref{lem:pop_eigval} and~\ref{lem:dev_h_0}, we have
\begin{align*}
|\widehat{\lambda}_2 - \lambda_1| 
&\le |\widehat{\lambda}_2 - \lambda_2| + |\lambda_1 - \lambda_2| \\
&\lesssim \frac{\sqrt{n}}{1 - 2p_\varepsilon} + \|\theta\|^2,
\end{align*}
with probability at least $1 - o(n^{-3})$. 
Under condition~\eqref{eq:imp_condn}, this implies
\[
|\widehat{\lambda}_2 - \lambda_1| \lesssim \|\theta\|^2
\quad \text{and} \quad
|\widehat{\lambda}_K| \lesssim \beta_n K^{-1} \|\theta\|^2,
\]
both holding with probability at least $1 - o(n^{-3})$.

Applying Lemma~D.3 of \citet{JIN2023}, we obtain
\[
\|\Xi_{*1}\Xi_{*1}^\top - \widehat{\Xi}_{*1}\widehat{\Xi}_{*1}^\top\| 
    \lesssim \frac{K\sqrt{n}}{(1 - 2p_\varepsilon)\|\theta\|^2},
\]
and
\[
\max_{t \in \{0,1\}}
\Big\|
\Xi^{(t)}_{*,2:K}(\Xi^{(t)}_{*,2:K})^\top
    - \widehat{\Xi}^{(t)}_{*,2:K}(\widehat{\Xi}^{(t)}_{*,2:K})^\top
\Big\|
    \lesssim
    \frac{K\sqrt{n}}{(1 - 2p_\varepsilon)\beta_n\|\theta\|^2},
\]
where for both $\Xi$ and $\widehat{\Xi}$, the superscript $t=0$ corresponds to eigenvectors associated with positive eigenvalues, and $t=1$ corresponds to those with negative eigenvalues.

Following the argument in the proof of Lemma~D.2 of \citet{JIN2023}, it follows that there exist $\omega \in \{-1, 1\}$ and an orthogonal matrix $O \in \mathbb{R}^{(K-1) \times (K-1)}$ such that
\[
\|\Xi_{*1} - \omega \widehat{\Xi}_{*1}\|_2 
    \lesssim \frac{K\sqrt{n}}{(1 - 2p_\varepsilon)\|\theta\|^2},
\]
and
\[
\|\Xi_{*,2:K} - \widehat{\Xi}_{*,2:K}O\|_F 
    \lesssim \frac{K^{3/2}\sqrt{n}}{(1 - 2p_\varepsilon)\beta_n\|\theta\|^2},
\]
with probability at least $1 - o(n^{-3})$.
\end{proof}
\subsection{Entrywise rate for the sample eigenvectors}

To characterize the estimation rate of $\widehat R_{i*}$ for each $i \in [n]$, it suffices to control the entrywise error of $\widehat\Xi_{*1}$ and the rowwise error of $\widehat\Xi_{*,2:K}$.  
In this subsection, we establish an $L_\infty$ perturbation bound for $\widehat\Xi_{*1}$; the corresponding rowwise bound for $\widehat\Xi_{*,2:K}$ follows from analogous arguments and is therefore omitted.

Starting from the eigenequation
\[
\widehat{\lambda}_1 \widehat{\Xi}_{*1}
= \sum_{k=1}^K \lambda_k \big(\widehat\Xi_{*1}^\top \Xi_{*k}\big)\,\Xi_{*k}
- \mathrm{diag}(\Omega)\,\widehat{\Xi}_{*1}
+ \pW\,\widehat{\Xi}_{*1},
\]
we extract the $i$th coordinate (noting that the diagonal term contributes $\Omega_{ii}\widehat\Xi_{i1}$):
\begin{equation}\label{eq:coord-decomp}
\widehat \Xi_{i1}
= \frac{\lambda_1}{\widehat\lambda_1}\big(\widehat \Xi_{*1}^\top \Xi_{*1}\big)\Xi_{i1}
+ \sum_{k=2}^K \frac{\lambda_k}{\widehat\lambda_1}\big(\widehat \Xi_{*1}^\top \Xi_{*k}\big)\Xi_{ik}
- \frac{\Omega_{ii}}{\widehat{\lambda}_1}\widehat \Xi_{i1}
+ \frac{e_i^\top \pW\,\widehat\Xi_{*1}}{\widehat\lambda_1}.
\end{equation}

By Theorem~D.3 of \citet{ke2022optimal}, there exists $\omega \in \{-1,1\}$ (as in Lemma~\ref{lem:l_2_perturb_eig_vect}) such that
\[
\big|\widehat \Xi_{*1}^\top \Xi_{*1}-\omega\big|
\;\lesssim\; \frac{K^2 n}{(1-2p_\varepsilon)^2\|\theta\|^4},
\qquad
\big\|\widehat \Xi_{*1}^\top \Xi_{*,2:K}\big\|_2
\;\lesssim\; \frac{K\sqrt{n}}{(1-2p_\varepsilon)\beta_n\|\theta\|^2}.
\]
Combining these with Weyl's inequality, Lemma~\ref{lem:pop_eigval}, Lemma~\ref{lem:dev_h_0} and condition~\eqref{eq:imp_condn}, and using $\theta_i=O(1)$ for all $i\in[n]$, we obtain with probability at least $1-o(n^{-3})$,
\begin{align}
\label{eq:diff_l_inty_prt}
\big|\widehat\Xi_{i1}-\omega\,\Xi_{i1}\big|
&\lesssim
\frac{|\widehat \lambda_1-\lambda_1|}{\widehat \lambda_1}\,|\Xi_{i1}|
\;+\;
\frac{\max_{k\ge2}\lambda_k}{\widehat \lambda_1}
\cdot
\frac{K\sqrt{n}}{(1-2p_\varepsilon)\beta_n\|\theta\|^2}
\,\|\Xi_{i,2:K}\|_2
\;+\;
\frac{e_i^\top\pW\,\widehat\Xi_{*1}}{\widehat\lambda_1}
\\[4pt]
&\lesssim
\frac{\sqrt{n}K}{(1-2p_\varepsilon)\beta_n\|\theta\|^2}\cdot \frac{\theta_i}{\|\theta\|}
\;+\;
\frac{e_i^\top\pW\,\widehat\Xi_{*1}}{\widehat\lambda_1}.
\end{align}
Hence, the remaining task is to control the fluctuation term
\[
\frac{e_i^\top\pW\,\widehat\Xi_{*1}}{\widehat\lambda_1}.
\]
A direct Bernstein bound is not immediately applicable since $\widehat\Xi_{*1}$ depends on $\pW$.  
We circumvent this dependence via a leave-one-out argument.  
Let $\pW^{(i)}$ be the matrix defined in~\eqref{Eq:leave_one_out_w}, and let $\widehat\Xi_{*1}^{(i)}$ be the leading eigenvector of the corresponding leave-one-out matrix.  
Decompose
\begin{align}
\label{eq:decomp_norm}
e_i^\top \pW\,\widehat\Xi_{*1}
= e_i^\top \pW\,\Xi_{*1}
\;+\;
e_i^\top \pW\,\big(\widehat\Xi_{*1}^{(i)}-\Xi_{*1}\big)
\;+\;
e_i^\top \pW\,\big(\widehat\Xi_{*1}-\widehat\Xi_{*1}^{(i)}\big).
\end{align}

For the first term,
\[
e_i^\top \pW\,\Xi_{*1} = \sum_{j\neq i}(\pW)_{ij}\Xi_{j1}.
\]
The variance of $\sum_{j\neq i}(\pW)_{ij}\Xi_{j1}$ satisfies
\[
\sum_{j\neq i}\mathrm{Var}\!\left((\pW)_{ij}\Xi_{j1}\right)
\lesssim
\sum_{j\neq i}\left\{\frac{1}{1+e^\varepsilon}+\theta_i\theta_j\right\}
\frac{\Xi_{j1}^2}{(1-2p_\varepsilon)^2}
\lesssim
\frac{1}{(1-2p_\varepsilon)^2},
\]
and
\[
\max_{i \neq j}|(\pW)_{ij}\Xi_{j1}| \lesssim \frac{\theta_j}{(1-2p_\varepsilon)\|\theta\|_2}.
\]
Applying Bernstein’s inequality yields along with \eqref{eq:imp_condn}, with probability at least $1-o(n^{-3})$,
\begin{align}
\label{eq:prtb_eqn_1}
\max_{i}\left|e_i^\top \pW\,\Xi_{*1}\right|
\lesssim
\frac{\sqrt{\log n}}{(1-2p_\varepsilon)}
+\frac{(\log n)\theta_j}{(1-2p_\varepsilon)\|\theta\|_2}
\lesssim
\frac{(\log n)}{(1-2p_\varepsilon)}.
\end{align}

Next, consider $e_i^\top \pW\,\big(\widehat\Xi_{*1}^{(i)}-\Xi_{*1}\big)$ for $i\in[n]$.  
By construction, $\widehat\Xi_{j1}^{(i)}-\Xi_{j1}$ is independent of $(\pW)_{ij}$ for all $i\neq j$, so Bernstein’s inequality applies again.  
We have
\[
\sum_{j\neq i}\mathrm{Var}\!\left((\pW)_{ij}\,\big(\widehat\Xi_{j1}^{(i)}-\Xi_{j1}\big)\right)
\lesssim
\sum_{j\neq i}\left\{\frac{1}{1+e^\varepsilon}+\theta_i\theta_j\right\}
\frac{\big(\widehat\Xi_{j1}^{(i)}-\Xi_{j1}\big)^2}{(1-2p_\varepsilon)^2}
\lesssim
\frac{1}{(1-2p_\varepsilon)^2},
\]
and
\[
\max_{i \neq j}\left|(\pW)_{ij}\big(\widehat\Xi_{j1}^{(i)}-\Xi_{j1}\big)\right|
\lesssim
\frac{1}{(1-2p_\varepsilon)1}
\,\big\|\widehat\Xi_{*1}^{(i)}-\Xi_{*1}\big\|_\infty.
\]
Thus, by Bernstein’s inequality,
\begin{align}
\label{eq:prtb_eqn_2}
\left|e_i^\top \pW\,\big(\widehat\Xi_{*1}^{(i)}-\Xi_{*1}\big)\right|
&\lesssim
\frac{\sqrt{\log n}}{(1-2p_\varepsilon)}
+(\log n)\frac{\big\|\widehat\Xi_{*1}^{(i)}-\Xi_{*1}\big\|_\infty}{(1-2p_\varepsilon)}\nonumber\\
&\lesssim
\frac{\sqrt{\log n}}{(1-2p_\varepsilon)}
+(\log n)\frac{\big\|\Xi_{*1}-\widehat\Xi_{*1}\big\|_\infty}{(1-2p_\varepsilon)}
+(\log n)\frac{\big\|\widehat\Xi_{*1}^{(i)}-\widehat\Xi_{*1}\big\|_2}{(1-2p_\varepsilon)}.
\end{align}
Next, we focus on bounding $\big\|\widehat\Xi_{*1}^{(i)}-\widehat\Xi_{*1}\big\|_2$. In that direction, observe that
\[
\big\|\pW-\pW^{(i)}\big\|
= \big\|e_i(\pW)_{i*}^\top + (\pW)_{i*}e_i^\top\big\|,
\]
since $(\pW)_{ii}=0$ for all $i \in [n]$.
Therefore, by Lemma 2 of \cite{abbe2020entrywise}, we have
\[
\big\|\widehat\Xi_{*1}^{(i)}-\widehat\Xi_{*1}\big\|_2 \lesssim \frac{\left\|(e_i(\pW)_{i*}^\top + (\pW)_{i*}e_i^\top)\widehat\Xi_{*1}\right\|_2}{\widehat{\lambda}_1-|\widehat{\lambda}_2|}.
\]
By the Weyl's inequality, Lemma \ref{lem:dev_h_0} and \eqref{eq:imp_condn}, we have
\[
\widehat{\lambda}_1-|\widehat{\lambda}_2| \gtrsim K^{-1}\|\theta\|^2_2.
\]
Therefore, we have
\begin{align}
\label{eq:decomp_leave_1}
    \big\|\widehat\Xi_{*1}^{(i)}-\widehat\Xi_{*1}\big\|_2 &\lesssim \frac{K}{\|\theta\|^2_2}\left\{\left|(\pW)_{i*}^\top\widehat \Xi_{*1}\right|+\left\|(\pW)_{i*}\right\|_2|\widehat \Xi_{i1}|\right\}\nonumber\\
    & \lesssim \frac{K}{\|\theta\|^2_2}\left\{\left|(\pW)_{i*}^\top\widehat \Xi_{*1}\right|+\left\|(\pW)_{i*}\right\|_2|\widehat \Xi_{i1}-\Xi_{i1}|+\left\|(\pW)_{i*}\right\|_2|\Xi_{i1}|\right\}.
\end{align}
By Bernstein’s inequality,
\[
\sum_{j\neq i}\left\{((\pW)_{ij})^2-\mathbb{E}[((\pW)_{ij})^2]\right\}
\lesssim
\frac{\sqrt{n}}{(1-2p_\varepsilon)}
+\frac{\log n}{(1-2p_\varepsilon)},
\]
which implies
\begin{align}
\label{eq:l_2_nrm_res_1}
\|(\pW)_{i*}\|_2
\lesssim
\frac{\sqrt{n}}{(1-2p_\varepsilon)},
\quad \text{with probability } 1-o(n^{-3}).
\end{align}
The above inequality, along with Lemma \ref{lem:true_eig_vect}
\begin{align}
\label{eq:leave_one_out_l_2_prtb}
    \big\|\widehat\Xi_{*1}^{(i)}-\widehat\Xi_{*1}\big\|_2 &\lesssim \frac{K}{\|\theta\|^2_2}\left\{\left|(\pW)_{i*}^\top\widehat \Xi_{*1}\right|+\frac{\sqrt{n}}{(1-2p_\varepsilon)}|\widehat \Xi_{i1}-\Xi_{i1}|+\frac{\sqrt{n}}{(1-2p_\varepsilon)}\frac{\theta_i}{\|\theta\|_2}\right\}.
\end{align}
Next, we focus on the first term. In that direction, we have
\begin{align}
    \left|(\pW)_{i*}^\top\widehat \Xi_{*1}\right| & \lesssim \left|(\pW)_{i*}^\top\Xi_{*1}\right|+\left|(\pW)_{i*}^\top\left(\widehat \Xi^{(i)}_{*1}-\Xi_{*1}\right)\right|+\left|(\pW)_{i*}^\top\left(\widehat \Xi^{(i)}_{*1}-\widehat\Xi_{*1}\right)\right|\\
    & \lesssim \left|(\pW)_{i*}^\top\Xi_{*1}\right|+\left|(\pW)_{i*}^\top\left(\widehat \Xi^{(i)}_{*1}-\Xi_{*1}\right)\right|+\frac{\sqrt{n}}{(1-2p_\varepsilon)}\left\|\widehat\Xi_{*1}^{(i)}-\widehat\Xi_{*1}\right\|_2.
\end{align}
By \eqref{eq:prtb_eqn_1} and \eqref{eq:prtb_eqn_2}, we have
\begin{align}
    &\left|(\pW)_{i*}^\top\Xi_{*1}\right|+\left|(\pW)_{i*}^\top\left(\widehat \Xi^{(i)}_{*1}-\Xi_{*1}\right)\right|\\
    & \lesssim \frac{(\log n)}{(1-2p_\varepsilon)}+(\log n)\frac{\big\|\Xi_{*1}-\widehat\Xi_{*1}\big\|_\infty}{(1-2p_\varepsilon)}
+(\log n)\frac{\big\|\widehat\Xi_{*1}^{(i)}-\widehat\Xi_{*1}\big\|_2}{(1-2p_\varepsilon)}.
\end{align}
Therefore,
\begin{align}
\label{eq:decomp_leave_2}
\left|(\pW)_{i*}^\top\widehat \Xi_{*1}\right| &\lesssim \frac{(\log n)}{(1-2p_\varepsilon)}+(\log n)\frac{\big\|\Xi_{*1}-\widehat\Xi_{*1}\big\|_\infty}{(1-2p_\varepsilon)}
+(\log n)\frac{\big\|\widehat\Xi_{*1}^{(i)}-\widehat\Xi_{*1}\big\|_2}{(1-2p_\varepsilon)}\\
& \quad +\frac{\sqrt{n}}{(1-2p_\varepsilon)}\left\|\widehat\Xi_{*1}^{(i)}-\widehat\Xi_{*1}\right\|_2\\
& \lesssim \frac{(\log n)}{(1-2p_\varepsilon)}+(\log n)\frac{\big\|\Xi_{*1}-\widehat\Xi_{*1}\big\|_\infty}{(1-2p_\varepsilon)}
+\frac{\sqrt{n}}{(1-2p_\varepsilon)}\big\|\widehat\Xi_{*1}^{(i)}-\widehat\Xi_{*1}\big\|_2.
\end{align}
Observe that since $\beta_n \in (0,1)$ by (2.3), we can conclude using \eqref{eq:imp_condn}, that
\[
\frac{K\sqrt{n}}{(1-2p_\varepsilon)\|\theta\|^2_2} \le \frac{K^3 (\log n)}{(1-2p_\varepsilon)\beta_n\sqrt{n}\widetilde{\theta}^2} \ll 1.
\]
This in turn implies
\[
\frac{K\sqrt{n}}{(1-2p_\varepsilon)\|\theta\|^2_2}\big\|\widehat\Xi_{*1}^{(i)}-\widehat\Xi_{*1}\big\|_2 \ll \big\|\widehat\Xi_{*1}^{(i)}-\widehat\Xi_{*1}\big\|_2.
\]
Plugging in \eqref{eq:decomp_leave_2} in \eqref{eq:decomp_leave_1} and using the above equation, Lemma \ref{lem:true_eig_vect} along with \eqref{eq:l_2_nrm_res_1}, we have
\begin{align}
\label{eq:prtb_eqn_3}
    \big\|\widehat\Xi_{*1}^{(i)}-\widehat\Xi_{*1}\big\|_2 &\lesssim \frac{K}{\|\theta\|^2_2}\Bigg\{\frac{(\log n)}{(1-2p_\varepsilon)}+(\log n)\frac{\big\|\Xi_{*1}-\widehat\Xi_{*1}\big\|_\infty}{(1-2p_\varepsilon)}\\
    &~~~~~~~~~~~~~~~+\frac{\sqrt{n}}{(1-2p_\varepsilon)}|\widehat \Xi_{i1}-\Xi_{i1}|+\frac{\sqrt{n}}{(1-2p_\varepsilon)}\frac{\theta_i}{\|\theta\|_2}\Bigg\}.
\end{align}
Combining \eqref{eq:decomp_norm} with \eqref{eq:prtb_eqn_1}-\eqref{eq:prtb_eqn_3} and \eqref{eq:imp_condn}, we can get
\begin{align}
e_i^\top \pW\,\widehat\Xi_{*1} &\lesssim \frac{(\log n)}{(1-2p_\varepsilon)}+(\log n)\frac{\big\|\Xi_{*1}-\widehat\Xi_{*1}\big\|_\infty}{(1-2p_\varepsilon)}\\
    &~~~~~~~~~~~~~~~+\frac{K}{\|\theta\|^2_2}\frac{\sqrt{n}}{(1-2p_\varepsilon)}|\widehat \Xi_{i1}-\Xi_{i1}|+\frac{\sqrt{n}}{(1-2p_\varepsilon)}\frac{\theta_i}{\|\theta\|_2}\frac{K}{\|\theta\|^2_2}.
\end{align}
This further implies using \eqref{eq:diff_l_inty_prt}, we have
\begin{align}
    \big|\widehat\Xi_{i1}-\omega\,\Xi_{i1}\big| & \lesssim \frac{K~(\log n)}{(1-2p_\varepsilon)\beta_n\|\theta\|^2_2}+\frac{K(\log n)}{\beta_n\|\theta\|^2_2}\frac{\big\|\Xi_{*1}-\widehat\Xi_{*1}\big\|_\infty}{(1-2p_\varepsilon)}+\frac{K\sqrt{n}}{(1-2p_\varepsilon)\beta_n\|\theta\|^2_2}\frac{\theta_i}{\|\theta\|_2}.
\end{align}
We can divide both sides by $\beta_n\|\theta\|^2_2(1-2p_\varepsilon)$, and take maximum over $i\in [n]$, to conclude that
\begin{align}
\frac{K(\log n)}{\beta_n\|\theta\|^2_2}\frac{\big\|\Xi_{*1}-\widehat\Xi_{*1}\big\|_\infty}{(1-2p_\varepsilon)} &\lesssim \frac{K(\log n)^{2}}{(1-2p_\varepsilon)^2\beta^2_n\|\theta\|^4}+\frac{K\sqrt{n}(\log n)}{(1-2p_\varepsilon)^2\beta^2_n\|\theta\|^4}\frac{\theta_i}{\|\theta\|^2}\\
& \lesssim \frac{K\,(\log n)}{(1-2p_\varepsilon)\beta_n\|\theta\|^2_2}+\frac{K\sqrt{n}}{(1-2p_\varepsilon)\beta_n\|\theta\|^2_2}\frac{\theta_i}{\|\theta\|_2},
\end{align}
where the last inequality follows by \eqref{eq:imp_condn}. Therefore,
\begin{align}
    \big|\widehat\Xi_{i1}-\omega\,\Xi_{i1}\big| & \lesssim \frac{K\,(\log n)}{(1-2p_\varepsilon)\beta_n\|\theta\|^2_2}+\frac{K\sqrt{n}}{(1-2p_\varepsilon)\beta_n\|\theta\|^2_2}\frac{\theta_i}{\|\theta\|_2}\\
    & \lesssim \frac{K\,(\log n)}{(1-2p_\varepsilon)\beta_n\|\theta\|^2_2}\left(1+\frac{\theta_i}{\widetilde \theta}\right), \quad \mbox{where $\widetilde{\theta}=\|\theta\|_2/\sqrt{n}$.}
\end{align}

Therefore, we get the following lemma.

\begin{lem}
\label{lem:l_infty_pertb_lead_suppl}
Suppose that $\theta_i = O(1)$ for all $i \in [n]$ and $\varepsilon \in [0, \varepsilon_0]$.  
Under Assumptions~\eqref{eq:singularity_b}, \eqref{eq:spectral_gap}, \eqref{eq:perron_root}, \eqref{eq:non_negligible_presence}, and condition~\eqref{eq:imp_condn}, there exist $\omega \in \{-1, 1\}$ (as in Lemma~\ref{lem:l_2_perturb_eig_vect}) such that with probability greater than $1-o(n^{-2})$, we have
\[
\big|\widehat\Xi_{i1}-\omega\,\Xi_{i1}\big|
\;\lesssim\;
\frac{K\,(\log n)}{(1-2p_\varepsilon)\beta_n\|\theta\|^2_2}\left(1+\frac{\theta_i}{\widetilde \theta}\right),
\]
where $\widetilde{\theta}=\|\theta\|_2/\sqrt{n}$.
\end{lem}

% \begin{proof}
% The claim follows by substituting~\eqref{eq:final_err_bd} into~\eqref{eq:diff_l_inty_prt} and invoking Lemma~\ref{lem:pop_eigval} together with~\eqref{eq:prtb_eig_val} and~\eqref{eq:imp_condn}.
% \end{proof}

Proceeding similarly, we can also show the following lemma.

\begin{lem}
\label{lem:l_infty_pertb_rest_suppl}
Suppose that $\theta_i = O(1)$ for all $i \in [n]$ and $\varepsilon \in [0, \varepsilon_0]$.  
Under Assumptions~\eqref{eq:singularity_b}, \eqref{eq:spectral_gap}, \eqref{eq:perron_root}, \eqref{eq:non_negligible_presence}, and condition~\eqref{eq:imp_condn}, there exist an orthogonal matrix $O \in \mathbb{R}^{(K-1)\times(K-1)}$ (as in Lemma~\ref{lem:l_2_perturb_eig_vect}) such that with probability greater than $1-o(n^{-2})$, we have
\[
\big\|\widehat\Xi_{i,2:K}-\Xi_{i,2:K}O\big\|_2
\;\lesssim\;
\frac{K^{3/2}\,(\log n)}{(1-2p_\varepsilon)\beta_n\|\theta\|^2_2}\left(1+\frac{\theta_i}{\widetilde \theta}\right),
\]
where $\widetilde{\theta}=\|\theta\|_2/\sqrt{n}$.
\end{lem}

\begin{proof}[Proof of Theorem \ref{thm:l_infty_pertb}]
    The proof of Theorem \ref{thm:l_infty_pertb} follows directly from Lemma \ref{lem:l_infty_pertb_lead_suppl} and \ref{lem:l_infty_pertb_rest_suppl}.
\end{proof}

\subsection{Perturbation of the SCORE normalized eigenvectors}
Let us consider the score normalized \cite{JIN2023} sample eigenvectors $\widehat{R}_{*2},\ldots,\widehat{R}_{*K}$ defined in \eqref{eq:score_norm_eig_vect}. In this section, we shall characterize the row-wise perturbation of the matrix $\widehat R$ from the similar matrix $R$ constructed by score-normalizing the population eigenvectors.

\begin{proof}[Proof of Corollary \ref{cor:concentration_of_R}]
    We can choose $\widehat{\Xi}$ in such a way such that $\omega$ in Lemma \ref{lem:l_infty_pertb_lead_suppl} is $1$. We shall only consider the nodes $i \in S_n(c_0)$. For these nodes, 
    \begin{equation}
    \label{eq:s_2_c_0}
    \theta_i \gtrsim \frac{K\,(\log n)}{(1-2p_\varepsilon)\beta_n\|\theta\|_2}.
    \end{equation}
    Therefore using \eqref{eq:imp_condn} and Lemma \ref{lem:true_eig_vect}, we have
\begin{align}
\label{eq:comp_eig_vect}
|\widehat{\Xi}_{i1}| \asymp \frac{\theta_i}{\|\theta\|_2}, \quad \mbox{and} \quad \|\widehat{\Xi}_{i,2:K}\| \asymp \frac{\sqrt{K}\theta_i}{\|\theta\|_2},
\end{align}
with probability greater than $1-o(n^{-2})$. By definition,
\begin{align}
\label{eq:prtb_r}
    \|\widetilde O\widehat{R}_{i*}-R_{i*}\|_2 & = \|\widehat{\Xi}_{i,2:K}\widetilde O/\widehat{\Xi}_{i,1}-\Xi_{i,2:K}/\Xi_{i,1}\|_2\\
    &\le \frac{\|\widehat{\Xi}_{i,2:K}\widetilde O-\Xi_{i,2:K}\|_2}{|\widehat\Xi_{i,1}|}+\|\Xi_{i,2:K}\|_2\bigg|\frac{1}{\widehat\Xi_{i1}}-\frac{1}{\Xi_{i1}}\bigg|.
\end{align}
Now, we consider the first term. Using \eqref{eq:comp_eig_vect}, Theorem \ref{thm:l_infty_pertb} and \eqref{eq:s_2_c_0}, we have
\begin{align}
    \frac{\|\widehat{\Xi}_{i,2:K}O_1-\Xi_{i,2:K}\|}{|\widehat\Xi_{i,1}|} & \lesssim \frac{K^{3/2}\,(\log n)}{(1-2p_\varepsilon)\beta_n\|\theta\|_2\theta_i}
    \left(1+\frac{\theta_i}{\widetilde{\theta}}\right),
\end{align}
where $\widetilde{\theta}=\|\theta\|_2/\sqrt{n}$. This further implies
\begin{align}
   \frac{\|\widehat{\Xi}_{i,2:K}O_1-\Xi_{i,2:K}\|}{|\widehat\Xi_{i,1}|} & \lesssim   \frac{K^{3/2}\,(\log n)}{(1-2p_\varepsilon)\beta_n\sqrt{n}~\widetilde{\theta}(\widetilde{\theta} \wedge \theta_i )}.
\end{align}
Similarly, we have
\begin{align}
     \|\Xi_{i,2:K}\|\bigg|\frac{1}{\widehat\Xi_{i1}}-\frac{1}{\Xi_{i1}}\bigg| & \lesssim C\frac{\|\Xi_{i,2:K}\|}{\Xi_{i1}}\frac{|\widehat\Xi_{i1}-\Xi_{i1}|}{\widehat\Xi_{i1}}\\
    & \lesssim \frac{K^{3/2}\,(\log n)}{(1-2p_\varepsilon)\beta_n\sqrt{n}~\widetilde{\theta}(\widetilde{\theta} \wedge \theta_i )},
\end{align}
simultaneously for all $i \in S_n(c_0)$ with probability greater than $1-o(n^{-2})$. Plugging the above inequalities in \eqref{eq:prtb_r} the corollary follows.
\end{proof}

\section*{Acknowledgement}
The authors used ChatGPT 5.0 to polish certain sections of the manuscript. However, AI-produced text was not directly used in any parts of the manuscript.

\bibliographystyle{abbrvnat}
\bibliography{Private_SBM}
\end{document}